\documentclass{statsoc}

\usepackage{amsmath}
\usepackage{amssymb}
\usepackage{multirow}
\usepackage{epsfig}
\usepackage{color}
\usepackage{booktabs,multirow}
\usepackage{hhline}
\usepackage{caption}
\usepackage{subcaption}
\usepackage{graphicx}
\usepackage{xr}

\def\rset{\mathbb R}

\def\Xset{\mathsf{X}} 
\def\Xset{\mathsf{X}} 
\def\Yset{\mathsf{Y}} 

\def\PP{\mathbb{P}} 
\def\PE{\mathbb{E}} 

 \newcommand{\seq}[1]{\left\langle#1\right\rangle}
\newcommand{\seqF}[1]{\left\langle#1\right\rangle_\textsf{F}}
\newcommand{\normF}[1]{\left\|#1\right\|_\textsf{F}}
\newcommand{\eqdef}{\ensuremath{\stackrel{\mathrm{def}}{=}}}

\newcommand*\samethanks[1][\value{footnote}]{\footnotemark[#1]}

\def\U{\mathcal{U}}
\def\M{\mathcal{M}}

\def\r{\textsf{r}}

\DeclareMathOperator*{\argmin}{\textsf{Argmin}}
\newcommand{\iid}{\stackrel{\mathrm{iid}}{\sim}}
\def\rmd{\mathrm{d}}

\newtheorem{assumption}{\textbf{H}\hspace{-3pt}}
\newtheorem{remark}{Remark}
\newtheorem{Lemma}{Lemma}
\newtheorem{theorem}{Theorem}

\newtheorem{algo}{Algorithm}
\newtheorem{assumptionA}{\textbf{A}\hspace{-3pt}}
\newtheorem{theoremA}{Theorem}

\usepackage{natbib}
\newcommand{\Tau}{\mathcal{T}}
\title{Change-Point Estimation in High-Dimensional Markov Random Field Models}
\author[Roy, Atchad\'e, Michailidis]{Sandipan Roy\samethanks[1], Yves Atchad\'e\samethanks[1]  and George Michailidis}
\coaddress{Department of Statistics, 439 West Hall, University of Michigan, Ann Arbor, MI 48109-1107, USA}
\email{sandipan@umich.edu, yvesa@umich.edu, gmichail@umich.edu}
\address{University of Michigan,
Ann Arbor, USA.}
\begin{document}
\begin{abstract}
This paper investigates a change-point estimation problem in the context of high-dimensional Markov random field models. Change-points represent a key feature in many dynamically evolving network structures. 
The change-point estimate is obtained by maximizing a profile penalized pseudo-likelihood function
under a sparsity assumption. We also derive a tight bound for the estimate, up to a logarithmic factor, even in
settings where the number of possible edges in the network far exceeds the sample size.
The performance of the proposed estimator is evaluated on synthetic data sets and is also used to explore voting patterns in the US Senate in the 1979-2012 period.
\end{abstract}

\section{Introduction}
Networks are capable of capturing dependence relationships and have been extensively employed in diverse scientific fields including biology, economics and the social sciences. A rich literature has been developed
for static networks leveraging advances in estimating sparse graphical models. However, increasing availability of data sets that evolve over time has accentuated the need for developing models for time
varying networks. Examples of such data sets include time course gene expression data, voting records of
legislative bodies, etc.

In this work, we consider modeling the underlying network through a Markov random field (MRF) that exhibits a change
in its structure at some point in time. Specifically, suppose we have $T$ observations 
$\left\{X^{(t)},1\leq t\leq T\right\}$ over $p$-variables with $X^{(t)}=\left(X_1^{(t)},\ldots,X_p^{(t)}\right)$ 
and $X_j^{(t)}\in\Xset$, for some finite set $\Xset$. Further, we assume that there exists a time point $\tau_\star=\lceil \alpha_\star T\rceil\in\{1,\ldots,T-1\}$, with $\alpha_\star\in(0,1)$, such that $\left\{X^{(t)},1\leq t\leq \tau_\star\right\}$ is an
independent and identically distributed sequence from a distribution $g_{\theta_\star^{(1)}}(\cdot)$ parametrized by a real symmetric matrix $\theta_\star^{(1)}$, while the remaining observations  $\left\{X^{(t)},\tau_\star+1\leq t\leq T\right\}$ forms also an
independent and identically distributed sequence from a distribution $g_{\theta_\star^{(2)}}(\cdot)$ parametrized by another real symmetric matrix $\theta_\star^{(2)}$. We assume that the two distributions $g_{\theta_\star^{(1)}}(\cdot)$, $g_{\theta_\star^{(2)}}(\cdot)$ belong to a parametric family of Markov random field distributions given by
\begin{equation}
\label{model:intro}
g_\theta(x)=\frac{1}{Z\left(\theta\right)}
\exp\left(\displaystyle\sum\limits_{j=1}^p \theta_{jj}B_0(x_j)+
\displaystyle\sum\limits_{1\leq k < j\leq p}\theta_{jk}B(x_j,x_k)\right),\;\; x\in\Xset^p,
\end{equation}
for a non-zero function $B_0:\;\Xset\to\rset$, and a non-zero symmetric function $B:\;\Xset\times\Xset\to\rset$ which encodes the interactions between the nodes. The term $Z\left(\mathbf{\theta}\right)$ is the corresponding normalizing constant. Thus,
the observations over time come from a MRF that exhibits a change in its structure at time $\tau_\star$ and
the matrices $\theta_\star^{(1)}$ and $\theta_\star^{(2)}$ encode the conditional independence structure between the $p$ random variables respectively before and after the change-point. 

The objective is to estimate the change-point $\tau_\star$, as well as the network structures $\theta_\star^{(1)}$ and $\theta_\star^{(2)}$.
Although the problem of identifying a change point has a long history in statistics (see \cite{bai}, \cite{carl}, \cite{hink1}, \cite{Load}, \cite{Lan}, \cite{mull}, \cite{rai} and references therein), its use in a high-dimensional network problem is novel and motivated by the US Senate voting record application discussed in Section 6. 
Note that in a low-dimensional setting, the results obtained for the change-point depend on the regime considered; specifically, if there is a fixed shift then the
asymptotic distribution of the change-point is given by the minimizer of a compound Poisson process (see \cite{kos}), while if the shift decreases to 0 as a function
of the sample size, the distribution corresponds to that of Brownian motion with triangular drift (see \cite{bhat}, \cite{mull}).

Note that the methodology developed in this paper is useful in other areas, where similar problems occur. Examples include biological settings, where a gene regulatory network may exhibit a significant change at a particular dose of a drug treatment, or in finance where major economic announcements may disrupt financial networks.

Estimation of time invariant networks from independent and identically distributed data based on the MRF model
has been a very active research area (see e.g. \cite{baner, hof, ravi, xu, gu} and references therein). Sparsity (an often realistic assumption) plays an important role in this literature, and allows the recovery of the underlying network with relatively few observations (\cite{ravi,gu}).

On the other hand, there is significant less work on time varying networks (see \cite{zhu}, \cite{kol1}, \cite{kol2} etc.). The closest setting to the current paper is the work in \cite{kol2}, which considers
Gaussian graphical models where {\em each} node can exhibit multiple change points. In contrast, this paper focuses on a {\em single} change-point impacting the global network structure of the underlying Markov
random field. In general, which setting is more appropriate depends on the application. In biological applications where the focus is on  particular biomolecules (e.g. genes, proteins, metabolites), nodewise change-point analysis would
typically be preferred, whereas is many social network applications
(such as the political network example considered below), global
structural changes in the network are of primary interest. Further, note that node-level changes detected at multiple nodes can be inconsistent, noisy and difficult to reconcile to extract global structural changes.

Another key difference between these two papers is the modeling framework employed. Specifically, in \cite{kol2} the number of nodes in the Gaussian graphical model is {\em fixed} and {\em smaller} than the available sample size.
The high-dimensional challenge comes from the possible presence of multiple change-points per node, which leads to a large number of parameters to be estimated. To overcome this issue, a total variation penalty is introduced, a strategy
that has worked well in regression modeling where the number of parameters is the same as the number of observations. On the other hand, this paper assumes a high-dimensional framework where the number of nodes (and hence the
number of parameters of interest, namely the edges) grow with the number of time points and focuses on estimating a single change-point in a general Markov random field model.

To avoid the intractable normalizing constant issue in estimating the network structures, we employ a pseudo-likelihood framework. As customary in the analysis of change-point problems (\cite{bai, Lan}), we employ a profile pseudo-likelihood function to obtain the estimate $\hat\tau$ of the true change-point $\tau_\star$. Under a sparsity assumption, and some regularity conditions that allow the number of parameters $p(p+1)$ to be much larger than the sample size $T$, we establish that with high probability, $|(\hat\tau/T)-\alpha_\star|=O(\log(pT)/T)$, as $p,T\to\infty$. Note that in classical change-point problems with a fixed-magnitude change, it is well-known that the maximum likelihood estimator of the change-point satisfies $|(\hat\tau/T)-\alpha_\star|=O_p(1/T)$ (see e.g. \cite{bhat}, \cite{bai}). This suggests that our result is rate-optimal, up to the logarithm factor $\log(T)$.
The derivation of the result requires a careful handling of model misspecification in Markov random fields as explained in Section 3, a novel aspect not present when estimating a single Markov random field from independent and identically distributed observations. See also \cite{atc} for another example of misspecification in Markov random fields. Further, to speed up the computation of the change-point estimator $\hat\tau$, we discuss a sampling strategy of the available observations, coupled with a smoothing procedure of the resulting likelihood function.

Last but not least, we employ the developed methodology to analyze the US Senate voting record from 1979 to 2012. In this application, each Senate seat represents a node of the network and the voting record of these $100$ Senate seats on a given bill is viewed as a realization of an underlying Markov random field that captures dependencies between them. The analysis strongly points to the presence of a change-point around January, 1995, the beginning of the tenure of the 104th Congress. This change-point comes at the footsteps of the November 1994 election that witnessed the Republican Party capturing the US House of Representatives for the first time since 1956. Other analyses based on more ad hoc methods, also point to a significant change occurring after the November 1994 election (e.g. \cite{moo}).

The remainder of the paper is organized as follows. Modeling assumptions and the estimation framework are
presented in Section 2, while Section 3 establishes the key technical results. Section 4 discusses 
computational issues and Section 5 evaluates the performance of the estimation procedure using synthetic data.
Section 6 illustrates the procedure on the US Senate voting record. Finally, proofs are deferred to the Supplement.
\section{Methodology}
Let $\{X^{(t)},\;1\leq t\leq T\}$ be a sequence of  independent random vector, where $X^{(t)}=(X_1^{(t)},\ldots,X_p^{(t)})$ is a $p$-dimensional Markov random field whose $j$-th component $X_j^{(t)}$ takes values in a finite set $\Xset$. We assume that there exists a time point (change point) $\tau_\star\in\{1,\ldots,T-1\}$ and symmetric matrices  $\theta^{(1)}_\star,\theta^{(2)}_\star\in\rset^{p\times p}$, such that for all $x\in\Xset^p$,
\[\PP\left(X^{(t)}=x\right) =g_{\theta_\star^{(1)}}(x),\;\;\mbox{ for }t=1,\ldots,\tau_\star,\]
and
\[\PP\left(X^{(t)}=x\right) =g_{\theta_\star^{(2)}}(x),\;\;\mbox{ for }t=\tau_\star+1,\ldots,T,\]
where $g_\theta$ is the Markov random field distribution given in (\ref{model:intro}). We assume without any loss of generality that $\tau_\star = \lceil \alpha_\star T\rceil$, for some $\alpha_\star\in (0,1)$, where $\lceil x\rceil$ denotes the smallest integer larger or equal to $x$. The likelihood function of the observations $\{X^{(t)},\;1\leq t\leq T\}$ is then given by 
\begin{equation}
\label{full:likelihood}
L_T\left(\tau,\theta^{(1)},\theta^{(2)}\vert X^{(1:T)}\right) =\prod_{t=1}^{\tau}g_{\theta^{(1)}}(X^{(t)})\prod_{t=\tau+1}^{T}g_{\theta^{(2)}}(X^{(t)}).
\end{equation} 

We write $\PE$ to denote the expectation operator with respect to $\PP$. For a symmetric matrix $\theta\in\rset^{p\times p}$, we write $\PP_\theta$ to denote the probability distribution on $\Xset^p$ with probability mass function $g_\theta$ and $\PE_\theta$ its expectation operator.

We are interested in estimating both the change point $\tau_\star$, as well as the parameters $\theta_\star^{(1)},\theta_\star^{(2)}$.  Let $\M_p$ be the space of all $p\times p$ real symmetric matrices. We equip $\M_p$ with the Frobenius inner product $\seqF{\theta,\vartheta}\eqdef\sum_{k\leq j}\theta_{jk}\vartheta_{jk}$, and the associated norm $\normF{\theta}\eqdef \sqrt{\seq{\theta,\theta}}$. This is equivalent to identifying $\M_p$ with the Euclidean space $\rset^{p(p+1)/2}$, and this identification prevails whenever we define gradients and Hessians of functions $f:\;\M_p\to\rset$. For $\theta\in\M_p$ we also define $\|\theta\|_1\eqdef \sum_{k\leq j}|\theta_{jk}|$, and $\|\theta\|_\infty\eqdef \sup_{k\leq j}|\theta_{jk}|$. If $u\in\rset^d$, for some $d\geq 1$, and $A$ is an ordered subset of $\{1,\ldots,d\}$, we define $u_A\eqdef (u_j,\,j\in A)$, and $u_{-j}$ is a shortcut for $u_{\{1,\ldots,d\}\setminus\{j\}}$.

To avoid some of the computational difficulties in dealing with the normalizing constant of $g_\theta$, we take a pseudo-likelihood approach. For $\theta\in\M_p$ and $j\in\left\{1,2,\ldots,p\right\}$, define $f_\theta^{(j)}(u\vert x)\eqdef\PP_\theta(X_j=u\vert X_{-j}=x_{-j})$, for $u\in\Xset$, and $x\in\Xset^p$.  From the expression of the joint distribution $g_\theta$ in (\ref{model:intro}), we have 
\begin{equation}\label{full:cond}
f_\theta^{(j)}(u\vert x)=\frac{1}{Z_\theta^{(j)}(x)}\exp\left(\theta_{jj}B_0(u) +\sum_{k\neq j} \theta_{jk}B(u,x_k)\right),\;u\in\Xset,\;x\in\Xset^p,\end{equation}
where
\begin{equation}\label{norm:const}
Z_\theta^{(j)}(x)\eqdef \int_{\Xset}\exp\left(\theta_{jj}B_0(z) +\sum_{k\neq j} \theta_{jk}B(z,x_k)\right)\rmd z.\end{equation}
The normalizing constant $Z_\theta^{(j)}(x)$ defined in (\ref{norm:const}) is actually a summation over $\Xset$, but for notational convenience we write it as an integral against the counting measure on $\Xset$. Next, we introduce
\begin{equation}\label{def:phi}
\phi(\theta,x)\eqdef -\sum_{j=1}^p \log f_\theta^{(j)}(x_j\vert x).\end{equation}
The negative log-pseudo-likelihood of the model (divided by $T$) is given by
\begin{equation}\label{log:ll}
\ell_T(\tau;\theta_1,\theta_2)\eqdef \frac{1}{T}\displaystyle\sum_{t=1}^\tau\phi(\theta_1,X^{(t)})+\frac{1}{T}\displaystyle\sum_{t=(\tau+1)}^T\phi(\theta_2,X^{(t)}).
\end{equation}
For $1\leq \tau<T$, and $\lambda>0$, we define the estimators
\[\widehat{\mathbf{\theta}}_{1,\tau}^{(\lambda)}\eqdef \argmin_{\theta\in\M_p}\frac{1}{T}\displaystyle\sum_{t=1}^\tau\phi(\theta,X^{(t)})+ \lambda\|\mathbf{\theta}\|_{1},\]
 and 
\[\widehat{\mathbf{\theta}}_{2,\tau}^{(\lambda)}\eqdef \argmin_{\theta\in\M_p}\frac{1}{T}\displaystyle\sum_{t=\tau+1}^T\phi(\theta,X^{(t)})+ \lambda\|\mathbf{\theta}\|_{1}.\]
We propose to estimate the change point $\tau_\star$ using a profile pseudo-likelihood approach. More precisely our estimator $\hat\tau$ is defined as
\begin{equation}\label{c_p_estimator}
\widehat{\tau}=\argmin_{\tau\in\mathcal{T}}\,\ell_T(\tau;\widehat{\mathbf{\theta}}_{1,\tau},\widehat{\mathbf{\theta}}_{2,\tau}),
\end{equation}
for a search domain $\mathcal{T}\subset\{1,\ldots,T\}$ of the form $\{k_l,k_l+1,\ldots,T-k_u\}$, where for each $\tau\in\mathcal{T}$, $\widehat{\mathbf{\theta}}_{1,\tau}=\hat\theta_{1,\tau}^{(\lambda_{1,\tau})}$ and $\widehat{\mathbf{\theta}}_{2,\tau}=\hat\theta_{1,\tau}^{(\lambda_{1,\tau})}$, for some positive penalty parameters $\lambda_{1,\tau}$, $\lambda_{2,\tau}$.  Since the network estimation errors at the boundaries of the time-line $\{1,\ldots,T\}$ are typically large, a restriction on the search domain is needed to guarantee the consistency of the method.  This motivates the introduction of $\mathcal{T}$. We give more details on $\mathcal{T}$ below. The penalty parameters $\lambda_{1,\tau}$ and $\lambda_{2,\tau}$ also play an important role in the behavior of the estimators, and we provide some guidelines below.

\section{Theoretical Results}
The recovery of $\tau_\star$ rests upon the ability of the estimators $\hat{\mathbf{\theta}}_{j,\tau}$ to correctly estimate $\theta_\star^{(j)}$, $j\in\left\{1,2\right\}$. 
Estimators for the static version of the problem where one has i.i.d. observations from a single Markov Random Field have been extensively studied; see
\cite{gu}, \cite{hof}, \cite{mein1}, \cite{ravi} and references therein for computational and theoretical details. However, in the present setting one of the estimators $\hat{\mathbf{\theta}}_{j,\tau}$, $j\in\left\{1,2\right\}$ is derived from a misspecified model. Hence, to establish the error bound  for $\|\hat\theta_{j,\tau}-\theta_\star^{(j)}\|_2$, we borrow from the approach in \cite{atc}. For penalty terms $\lambda_{j,\tau}$ as in (\ref{lambda1:lambda2}) and under some regularity assumptions, we derive a bound on the estimator errors $\|\hat\theta_{j,\tau}-\theta_\star^{(j)}\|_2$, for all $\tau\in\mathcal{T}$. We then use this result to show that the profile pseudo-log-likelihood estimator $\hat\tau$ is an approximate minimizer of $\tau\mapsto \ell_T(\tau;\theta_\star^{(1)},\theta_\star^{(2)})$ and this allows us to establish a bound on the distance between $\hat\tau$ and the true change point $\tau_\star$.

We assume that the penalty parameters take the following specific form. 
\begin{equation}
\label{lambda1:lambda2}
\lambda_{1,\tau}=\frac{32c_0\sqrt{\tau\log\left(dT\right)}}{T}\mbox{ and } \lambda_{2,\tau}=\frac{32c_0\sqrt{\left(T-\tau\right)\log\left(dT\right)}}{T},
\end{equation}
where $d\eqdef p(p+1)/2$, and 
\begin{equation}\label{def:c0}
c_0=\sup_{u,v\in \Xset}|B_0(u)-B_0(v)|\vee \sup_{x,u,v\in \Xset}|B(x,u)-B(x,v)|,\end{equation}
which serves as (an upper bound on the) standard deviation of the random variables $B_0(X)$, $B(X,Y)$. In practice,  we use $\lambda_{1,\tau}=a_1T^{-1}c_0\sqrt{\tau\log(dT)}$, and  $\lambda_{2,\tau}=a_2T^{-1}c_0\sqrt{(T-\tau)\log(dT)}$, where $a_1, a_2$ are chosen from the data by an analogue of the Bayesian Information Criterion (\cite{scw}).

For $j=1,2$, define $\mathcal{A}_j\stackrel{\text{def}}{=}\left\{1\leq k\leq i\leq p:\theta_{\star ik}^{(j)}\neq 0\right\}$, and define $s_j\eqdef |\mathcal{A}_j|$ the cardinality (and hence the sparsity) of the true model parameters. We also define 
 \begin{equation}\label{defCj}
 \mathbb{C}_{j}\eqdef \left\{\mathbf{\theta}\in\mathcal{M}_p:\displaystyle\sum_{\left(k,i\right)\in\mathcal{A}_j^c}|\theta_{ik}^{(j)}|\leq 3\displaystyle\sum_{\left(k,i\right)\in\mathcal{A}_j}|\theta_{ik}^{(j)}|\right\}\mbox{, $j\in\left\{1,2\right\}$},\end{equation}
used next in the definition of the restricted strong convexity assumption.

\begin{assumption}\label{H1} {\em [Restricted Strong Convexity]}
For $j\in\{1,2\}$, and $X\sim g_{\theta^{(j)}_\star}$, there exists $\rho_j>0$ such that for all $\Delta\in\mathbb{C}_j$,
\begin{equation}\label{eqH1}
\sum_{i=1}^p \PE_{\theta_\star^{(j)}}\left[\textsf{Var}_{\theta_\star^{(j)}}\left(\sum_{k=1}^p\Delta_{ik}B_{ik}(X_i,X_k)\vert X_{-i}\right)\right]\geq 2\rho_j \,\|\Delta\|^2_2,\end{equation}
where $B_{ik}(x,y)=B_0(x)$ if $i=k$, and $B_{ik}(x,y)=B(x,y)$ if $i\neq k$.
\end{assumption}

\begin{remark}\label{rem:H1}
Assumption H\ref{H1} is a (averaged) restricted strong convexity (RSC) assumption on the negative log-pseudo-likelihood function $\phi(\theta,x)$. This can be seen by noting that (\ref{eqH1}) can also be written as
\[\Delta'\PE\left[\nabla^{(2)}\phi(\theta_\star^{(j)},X^{(j)})\right]\Delta\geq 2\rho_j\|\Delta\|^2_2,\;\;\;X^{(j)}\sim g_{\theta_\star^{(j)}},\;\;\Delta\in\mathbb{C}_j,\;\;j\in\{1,2\}.\]
These restricted strong convexity assumptions of objective functions are more pertinent in high-dimensional problems and appear in one form or another in the analysis of high-dimensional statistical methods (see e.g. \cite{negh} and references therein). Note that the RSC assumption is expressed here in expectation, unlike \cite{negh} which uses an almost sure version. Imposing this  assumption in expectation (that is, at the population level) is more natural, and is known to imply the almost sure version in many instances (see \cite{rudelson:zhou:13}, and Lemma 4 in the Supplement).
\end{remark}
We impose the following condition on the change point and the sample size.
\begin{assumption}\label{H2} {\em [Sample size requirement]}
We assume that there exists $\alpha_\star\in\left(0,1\right)$ such that $\tau_\star=\lceil\alpha_\star T\rceil\in\{1,\ldots,T-1\}$, and the sample size $T$ satisfies 
\[\min\left(\frac{T}{2^{11}\log(pT)},\frac{T}{48^2\times 32^2\log\left(dT\right)}\right) \geq c_0^2\max\left(\frac{s^2_1}{\alpha_\star\rho^2_1},\frac{s^2_2}{\left(1-\alpha_\star\right)\rho^2_2}\right),\]
where $\rho_1$, and $\rho_2$ are as in H\ref{H1}.
\end{assumption}

\begin{remark}
Note that the constants $2^{11}$ and $48^2\times 32^2$ required in H\ref{H2} will typically yield a very conservative bound on the sample size $T$. We believe these large constants are mostly artifacts of our techniques, and can be improved. The key point of H\ref{H2} is the fact that we require the sample $T$ to be such that $T/\log(T)$ is a linear function of $\max(s_1^2,s_2^2)\log(p)$. Up to the $\log(T)$ term, this condition is in agreement with recent results on high-dimensional sparse graphical model recovery.
\end{remark}

The ability to detect the change-point requires that the change from $\theta_\star^{(1)}$ to $\theta_\star^{(2)}$ be identifiable. 

\begin{assumption}\label{H3} {\em [Identifiability Condition]}
Assume that
$\theta_\star^{(1)}\neq \theta_\star^{(2)}$, and
\begin{equation}
\label{H3:SNR}
\kappa\eqdef \min\left(\PE_{\theta_\star^{(2)}}\left[\phi(\theta_\star^{(1)},X)-\phi(\theta_\star^{(2)},X)\right],\PE_{\theta_\star^{(1)}}\left[\phi(\theta_\star^{(2)},X)-\phi(\theta_\star^{(1)},X)\right]\right)> 0.
\end{equation}
\end{assumption}

\begin{remark}
Assumption H\ref{H3} is needed for the identifiability of the change-point $\tau_\star$. Since the distributions $g_\theta$ are discrete data analogs of Gaussian graphical distributions, it is informative to look at H\ref{H3} for Gaussian graphical distributions. Indeed, if $g_\theta$ is the density of the $p$-dimensional normal distribution $\textbf{N}(0,\theta^{-1})$ with precision matrix $\theta$, and if we take $\phi(\theta,x)=-\log g_\theta(x)$, then it can be easily shown that
\[ \kappa \geq \frac{1}{4L^2}\|\theta_\star^{(2)}-\theta_\star^{(1)}\|_2^2,\]
where $L$ is an upper bound on the largest eigenvalue of $\theta_\star^{(1)}$ and $\theta_\star^{(2)}$. Hence in this case H\ref{H3} holds. Such a general result is more difficult to establish for discrete Markov random fields. However, it can be easily shown that H\ref{H3} holds if 
\begin{multline}\label{eq:H3}
\left(\theta_\star^{(1)}-\theta_\star^{(2)}\right)'\PE_{\theta_\star^{(2)}}\left[\nabla^{(2)}\phi(\theta_\star^{(2)},X)\right]\left(\theta_\star^{(1)}-\theta_\star^{(2)}\right)' > 0,\\
\mbox{ and }\;\;  \left(\theta_\star^{(2)}-\theta_\star^{(1)}\right)'\PE_{\theta_\star^{(1)}}\left[\nabla^{(2)}\phi(\theta_\star^{(1)},X)\right]\left(\theta_\star^{(2)}-\theta_\star^{(1)}\right)' > 0.\end{multline}
And in the particular setting where  $\theta_\star^{(1)}$ and $\theta_\star^{(2)}$ have similar sparsity patterns (in the sense that $\theta_\star^{(2)}-\theta_\star^{(1)}\in\mathbb{C}_1\cap\mathbb{C}_2$), then  (\ref{eq:H3}) follows from H\ref{H1}, and the discussion in Remark \ref{rem:H1}.

\end{remark}

Finally, we define the search domain as the set 
\begin{equation}\label{def:Tau}
\Tau=\Tau_+\cup\Tau_-,\end{equation}
 where $\Tau_+$ is defined as the set of all time-points $\tau\in\left\{\tau_\star+1,\ldots,T\right\}$ such that
\begin{equation}\label{cond H2:Tau_p}
c_0b(\tau-\tau_\star)\leq 2\sqrt{\tau\log(dT)},\;\;\;\mbox{ and }\;\;\; 64c_0^3bs_1(\tau-\tau_\star)\leq \rho_1\tau,
\end{equation}
and $\Tau_-$ is defined as the set of all time-point $\tau\in\left\{1,\ldots,\tau_\star\right\}$ such that
\begin{equation}
\label{cond H2:Tau_m}
c_0b(\tau_\star-\tau)\leq 2\sqrt{(T-\tau)\log(dT)},\;\;\;\mbox{ and }\;\;\; 64c_0^3bs_2(\tau_\star-\tau)\leq \rho_2(T-\tau),
\end{equation}
where
\begin{equation}
\label{def:b}
b\stackrel{\text{def}}{=}\displaystyle\sup_{1\leq j\leq p}\sum_{k=1}^p\big|\theta_{\star jk}^{(2)}-\theta_{\star jk}^{(1)}\big|\mbox{.}
\end{equation}

Furthermore, for all $\tau\in\mathcal{T}$, 
\begin{multline}\label{cond H2:Tau_tau}
\tau \geq \max\left(2^{11},(48\times 32)^2\right)c_0^2\left(\frac{s_1}{\rho_1}\right)^2\log(dT),\;\;\;\\
\mbox{ and }\;\;\; T-\tau \geq \max\left(2^{11},(48\times 32)^2\right)c_0^2\left(\frac{s_2}{\rho_2}\right)^2\log(dT).
\end{multline}

\begin{remark}
Notice that $\mathcal{T}$ is of the form $\{k_l,k_l+1,\ldots,\tau_\star,\tau_\star+1,\ldots,T-k_u\}$, since for $\tau$ close to $\tau_\star$ both (\ref{cond H2:Tau_p}), (\ref{cond H2:Tau_m}), and (\ref{cond H2:Tau_tau}) hold provided that $T$ is large enough. 
\end{remark}

\medskip
We can then establish the key result of this paper. Set
\[M = \left[\frac{s_1}{\rho_1}\left(1+c_0\frac{s_1}{\rho_1}\right) + \frac{s_2}{\rho_2}\left(1+ c_0\frac{s_2}{\rho_2}\right)\right].\]

\begin{theorem}\label{thm1}
Consider the model posited in (\ref{full:likelihood}), and assume H\ref{H1}-H\ref{H3}. Let $\hat \tau$ be the estimator defined in (\ref{c_p_estimator}), with $\lambda_{1,\tau},\lambda_{2,\tau}$ as in (\ref{lambda1:lambda2}), and with a search domain $\mathcal{T}$ that satisfies (\ref{cond H2:Tau_p}), (\ref{cond H2:Tau_m}), and (\ref{cond H2:Tau_tau}). Then there exists a universal finite constant $a>0$, such that with $\delta = aMc_0^2\log(dT)$, we have 
\begin{equation}\label{bound:thm1}
\PP\left(\left|\frac{\widehat{\tau}}{T}-\alpha_\star\right|>\frac{4\delta}{\kappa T}\right) \leq \frac{16}{d} + \frac{4\exp\left(-\frac{\delta}{32c_0^2s}\left(\frac{\kappa}{\|\theta_\star^{(2)} - \theta_\star^{(1)}\|^2_2}\right)^2\right)}{1-\exp\left(-\frac{\kappa^2}{2^7c_0^2s\|\theta_\star^{(2)} - \theta_\star^{(1)}\|^2_2}\right)},
\end{equation}
where $s$ is the number of non-zero components of $\theta_\star^{(2)} - \theta_\star^{(1)}$.
\end{theorem}

\medskip
Theorem \ref{thm1} gives a theoretical guarantee that for large $p$ and for large enough sample size $T$ such that $(T/\log(T))=O(\max(s_1^2,s_2^2)\log(p))$, $|\hat\tau/T-\alpha_\star|=O(\log(pT)/T)$ with high-probability. For fixed-parameter change-point problems, the maximum likelihood estimator of the change-point is known to satisfy $|\hat\tau/T-\alpha_\star|=O_P(1/T)$ (see e. g. \cite{bai}). This shows that our result is rate-optimal, up to the logarithm factor $\log(T)$. Whether one can improve the bound and remove the $\log(T)$ term hinges on the existence of an exponential bound for the maximum of weighted partial sums of sub-Gaussian random variables, as we explain in Remark 1 of the Supplement. Whether such bound holds is currently an open problem, to the best of our knowledge. However, note that the $\log(p)$ term that appears in the theorem cannot be improve in general in the large $p$ regime.

If the signal $\kappa$ introduced in H\ref{H3} satisfies
\begin{equation}\label{lb:kappa}
\kappa\geq \kappa_0\|\theta_\star^{(2)}-\theta_\star^{(1)}\|_2^2,\end{equation}
then the second term on right-hand side of (\ref{bound:thm1}) is upper bounded by
\begin{equation}\label{lb:kappa:2} \left(\frac{1}{dT}\right)^{\frac{aM\kappa_0}{32s}}\frac{1}{1-\exp\left(-\frac{\kappa_0^2}{2^7c_0^2s}\|\theta_\star^{(2)}-\theta_\star^{(1)}\|_2^2\right)}.
\end{equation}
This shows that Theorem \ref{thm1} can also be used to analyze cases where $\|\theta_\star^{(2)}-\theta_\star^{(1)}\|_2^2\downarrow 0$, as $p\to\infty$. In such cases, consistency is guaranteed provided that  the term in (\ref{lb:kappa:2}) converges to zero. From the right-hand side of (\ref{lb:kappa}), we then see that the convergence rate of the estimator in such cases is changed to
\[\frac{c_0^2}{\|\theta_\star^{(2)}-\theta_\star^{(1)}\|_2^2}\frac{\log(dT)}{T}.\]

Another nice feature of Theorem \ref{thm1} is the fact that the constant $M$ describes the behavior of the change-point estimator as a function of the key parameters of the problem. In particular, the bound in (\ref{bound:thm1}) shows that the change-point estimator improves as $s_1,s_2$ (the number of non-zero entries of the matrices $\theta_\star^{(1)},\theta_\star^{(2)}$ resp.), or the noise term $c_0$ (the maximum fluctuation of $B_0$ and $B$) decrease. 

\section{Algorithm and Implementation Issues}
Given a sequence of observed $p$-dimensional vectors $\{x^{(t)}, 1\leq t\leq T\}$, we propose the following algorithm to compute the change point $\hat\tau$, as well as the estimate the estimates $\bigl(\hat{\theta}_{1,\hat\tau}, \hat{\theta}_{2,\hat\tau} \bigr)$.

\begin{algo}[\textbf{Basic Algorithm}]\;\;\;Input:\; a sequence of observed $p$-dimensional vectors $\{x^{(t)}, 1\leq t\leq T\}$, and $\Tau\subseteq \{1,\ldots,T\}$ the search domain.
\begin{enumerate}
\item For each $\tau\in \Tau$, estimate $\hat{\theta}_{1,\tau}, \hat{\theta}_{2,\tau} $ using for instance the algorithm in \cite{hof}.
\item
For each $\tau\in \Tau$, plug-in
the estimates $\hat{\theta}_{1,\tau}, \hat{\theta}_{2,\tau} $ in (\ref{log:ll}) and obtain the profile (negative) pseudo-log-likelihood function
$\mathcal{P}\ell(\tau)\eqdef \ell_T(\tau;\hat{\theta}_{1,\tau},\hat{\theta}_{2,\tau})$.
\item
Identify $\hat\tau$ that achieves the minimum of $\mathcal{P}\ell(\tau)$ over the grid $\Tau$, and use $\hat{\theta}_{1,\hat\tau}, \hat{\theta}_{2,\hat\tau}$ as  the  estimates of $\theta_\star^{(1)}$ and $\theta_\star^{(2)}$, respectively.
\end{enumerate} 
\end{algo}

In our implementation of the Basic Algorithm, we choose a search domain $\Tau$ of the form $\Tau=\left\{k_l, k_l+1,\ldots,T-k_l\right\}$, with $k_l$ sufficiently large to ensure reasonably good estimation errors at the boundaries. Existing results (\cite{ravi,gu}) suggest that a sample size of order $O(s^2\log(d))$ is needed, where $s$ is the number of edges, for a good recovery of Markov random fields.
\\\\Note that to identify the change-point $\hat\tau$ the algorithm requires a 
{\em full scan} of all the time points in the set $\Tau$, which can be expensive when $\Tau$ is large. As a result, we propose a fast implementation that operates in two stages.  In the first stage, a coarser grid $\Tau_1\subset \Tau$
of time points is used and steps (a) and (b) of the Basic Algorithm are used to obtain $\ell_T(\tau;\hat{\theta}_{1,\tau},\hat{\theta}_{2,\tau}), \tau\in \Tau_1$.
Subsequently, the profile likelihood function $\ell_T$ is smoothed using a Nadaraya-Watson kernel (\cite{nad}). Based on this smoothed version of the profile
likelihood, an initial estimate of the change-point is obtained. In the second stage, a new fine-resolution grid $\Tau_2$ is formed around the first stage
estimate of $\hat\tau$. Then, the Basic Algorithm is used for the grid points in $\Tau_2$ to obtain the final estimate. This leads to a more practical algorithm summarized next.

\begin{algo}[\textbf{Fast Implementation Algorithm}]\;\;\;Input:\;a sequence of observed $p$-dimensional vectors $\{x^{(t)}, 1\leq t\leq T\}$, and $\Tau\subseteq \{1,\ldots,T\}$ the search domain.
\begin{enumerate}
\setlength{\itemsep}{5pt}
\item Find a coarser grid $\Tau_1$ of time points.
\item For each $\tau\in \Tau_1$,  use steps (a) and (b) of the Basic Algorithm to obtain 
	$\mathcal{P}\ell_T(\tau), \ \ \tau\in\Tau_1$.
\item Compute the profile negative pseudo-log-likelihood over the interval $[1, T]$ by Nadaraya-Watson kernel smoothing:
	\[\widetilde{\mathcal{P}\ell_{1s}}(\tau)\eqdef \frac{\sum_{\tau_i\in \Tau_1}K_{h_\nu}\left(\tau,\tau_i\right)\ell(\tau_i;\widehat{\mathbf{\theta}}_{1,\tau_i},\widehat{\mathbf{\theta}}_{2,\tau_i})}{\sum_{\tau_i\in \Tau_1}\ell\left(\tau_i;\widehat{\mathbf{\theta}}_{1,\tau_i},\widehat{\mathbf{\theta}}_{2,\tau_i}\right)},\;1\leq \tau\leq T.\]
	The first stage change-point estimate is then obtained as 
\[	\widehat{\tau}=\displaystyle\argmin_{1<\tau<T}\widetilde{\mathcal{P}\ell_{1s}}(\tau).\]
\item Form a second stage grid $\Tau_2$ around the first stage estimate $\hat{\tau}$ and for each $\tau\in\Tau_2$, estimate $\widehat{\widehat{\mathbf{\theta}}}_{1,\tau}$ and $\widehat{\widehat{\mathbf{\theta}}}_{2,\tau}$  using steps (a) and (b) of the Basic Algorithm.
\item Construct the second stage smoothed profile pseudo-likelihood 
\[\widetilde{\mathcal{P}\ell_{2s}}(\tau)\eqdef \frac{\sum_{\tau_i\in \Tau_2}K_{h_\nu}\left(\tau,\tau_i\right)\ell\left(\tau_i;\widehat{\widehat{\mathbf{\theta}}}_{1,\tau_i},\widehat{\widehat{\mathbf{\theta}}}_{2,\tau_i}\right)}{\sum_{\tau_i\in \Tau_2}\ell\left(\tau_i;\widehat{\widehat{\mathbf{\theta}}}_{1,\tau_i},\widehat{\widehat{\mathbf{\theta}}}_{2,\tau_i}\right)},\;\min(\Tau_2)\leq \tau\leq \max(\Tau_2).\]
	The final change-point estimate is then given by 
\[\widehat{\widehat{\tau}}=\displaystyle\argmin_{\min(\Tau_2)\leq \tau\leq \max(\Tau_2)}\;\widetilde{\mathcal{P}\ell_{2s}}(\tau).\]
\end{enumerate}
\end{algo}

\section{Performance Assessment}
\subsection{Comparing Algorithm 1 and Algorithm 2}
We start by examining the relative performance of both the Basic (Algorithm 1) and the Fast Implementation Algorithms (Algorithm 2). We use the so called Ising model; i.e. when ~\eqref{model:intro} has $B_0\left(x_j\right)=x_j$, $B\left(x_j,x_k\right)=x_jx_k$ and $\Xset\equiv\left\{0,1\right\}$. In all simulation setting the sample size is set to $T=700$, and the true change-point is at $\tau_\star=350$, while the network size $p$ varies from 40-100. All the simulation results reported below are based on 30 replications of Algorithm 1 and Algorithm 2.

The data are generated as follows. We first generate two $p\times p$ symmetric adjacency matrices each having density 10\%; i.e. only $\sim$10\% of the entries are different than zero. Each off-diagonal element of $\mathbf{\theta}_{\star jk}^{(i)}$, ($i=1,2$) is drawn uniformly from $\left[-1,-0.5\right]\cup\left[0.5,1\right]$ if there is an edge between nodes $j$ and $k$, otherwise $\mathbf{\theta}_{\star jk}^{(i)}=0$. All the diagonal entries are set to zero. Given the two matrices $\mathbf{\theta}_{\star }^{(1)}$ and $\mathbf{\theta}_{\star }^{(2)}$, we generate the data $\left\{X^{(t)}\right\}_{t=1}^{\tau_\star}\iid g_{{\mathbf{\theta}}_{*}^{(1)}}$ and $\left\{X^{(t)}\right\}_{t=\tau_\star+1}^{T}\iid g_{{\mathbf{\theta}}_{*}^{(2)}}$ by Gibbs sampling.

 Different ``signal strenghts" are considered, by setting the degree of similarity between $\theta_\star^{(1)}$ and $\theta_\star^{(2)}$ to $0\%$, $20\%$ and $40\%$. The degree of similarity is the proportion of equal off-diagonal elements between $\theta_\star^{(1)}$ and $\theta_\star^{(2)}$. Thus, the difference $\|\theta_\star^{(2)}-\theta_\star^{(1)}\|_1$ becomes smaller for higher degree of similarity and as can be seen from Assumption H3, the estimation problem becomes harder in such cases. 
 
 The choice of the tuning parameters $\lambda_{1,\tau}$ and $\lambda_{2,\tau}$ were made based on Bayesian Information Criterion (BIC) where we search $\lambda_{1,\tau}$ and $\lambda_{2,\tau}$ over a grid $\Lambda$ and for each penalty parameter the $\lambda$ value that minimizes the BIC score (defined below)  over $\Lambda$ is selected. If we define $\lambda_1^{BIC}$ and $\lambda_2^{BIC}$ as the selected $\lambda$ values for $\lambda_1$ and $\lambda_2$ by BIC we have
$$\lambda_1^{BIC}=\argmin_{\lambda\in\Lambda}-\frac{2}{T}\displaystyle\sum_{t=1}^\tau\phi\left(\hat\theta_{1,\tau}^{(\lambda)},X^{(t)}\right)+\log(\tau)\|\hat\theta_{1,\tau}^{(\lambda)}\|_0\mbox{  and }$$
$$\lambda_2^{BIC}=\argmin_{\lambda\in\Lambda}-\frac{2}{T}\displaystyle\sum_{t=\tau+1}^T\phi\left(\hat\theta_{2,\tau}^{(\lambda)},X^{(t)}\right)+\log(T-\tau)\|\hat\theta_{2,\tau}^{(\lambda)}\|_0$$ where $\|\theta\|_0 \eqdef \sum_{k\leq j}\textbf{1}_{\{|\theta_{jk}|>0\}}$.

For the fast algorithm (Algorithm 2), the first stage grid employed had a step size of 10 and
ranged from 60 to 640, while the second stage grid was chosen in the interval $[\hat\tau-30, \hat\tau+30]$ with
a step-size of 3. 

We present the results for Algorithm 1 in Table \ref{basic_result} for the case $p=40$. It can be seen that Algorithm 1 performs very well for stronger signals (0\% and 20\% similarity), while  there is a small degradation for the 40\% similarity setting. The results on the specificity, sensitivity and the relative error of the estimated network structures  are given in Table ~\ref{basic_estnet}. Specificity is defined as the proportion of true negatives and can also be interpretated as (1-Type 1 error). On the other hand sensitivity is the proportion of true positives and can be interpreted as the power of the method. The results for Algorithm 2 for $p=40,60$ and $p=100$, for the change-point estimates are given in Table~\ref{estimresult}, while
the specificity, sensitivity and relative error of the estimated network structures are given in Table~\ref{matresult}.  These results show that Algorithm 2 has about $20\%$ higher mean-squared error (MSE) compared to Algorithm 1. However as pointed out in Section 4, Algorithm 2 is significantly faster. In fact in this particular simulation setting, Algorithm 2 is almost 5 times faster in a  standard computing environment with 4 CPU cores. See also the results in Table~\ref{runtab} which reports the ratio of the run-time of a single iteration of Algorithm 1 and Algorithm 2.

Further, selected plots of the profile smoothed pseudo-log-likelihood functions $\widetilde{\mathcal{P}\ell_{1s}}(\tau)$ and $\widetilde{\mathcal{P}\ell_{2s}}(\tau)$ from the first and second stage of Algorithm 2  are given in Figure~\ref{fig1}.

\vspace{1.5cm}

\begin{center}
\captionof{table}{Change-point estimation results using the Basic Algorithm, for different percentages of similarity.\label{basic_result}}  
\begin{tabular}{l c c c c}  
\hline\hline                       
 $p$ & \% of Similarity & $\widehat{\tau}$ & RMSE & CV 
\\ [0.5ex]   
\hline
\multirow{3}{*}{40} & 0 & 355 & 14.77 & 0.03\\
                    & 20 & 362 & 24.65 & 0.06\\
										& 40 & 375 & 38.49 & 0.08\\ \hline
\end{tabular}
\end{center}

\vspace{0.5cm}

\begin{center}
\captionof{table}{Specificity, sensitivity and relative error in estimating $\theta_\star^{(1)}$ and $\theta_\star^{(2)}$ from the Basic Algorithm, with different percentages of similarity.\label{basic_estnet}}  
\begin{tabular}{llllllll}
\hline\hline 
 $p$ & \% of Similarity & \multicolumn{2}{c}{Specificity} & \multicolumn{2}{c}{Sensitivity} & \multicolumn{2}{c}{Relative error} \\
& & $\mathbf{\theta}_*^{(1)}$ & $\mathbf{\theta}_*^{(2)}$ & $\mathbf{\theta}_*^{(1)}$ & $\mathbf{\theta}_*^{(2)}$ & $\mathbf{\theta}_*^{(1)}$ & $\mathbf{\theta}_*^{(2)}$\\
    \midrule
    \multirow{3}{*}{40} & 0 & 0.78 & 0.87 & 0.79 & 0.89 & 0.70 & 0.63\\
                    & 20 & 0.74 & 0.88 & 0.80 & 0.88 & 0.72 & 0.67\\
										& 40 & 0.71 & 0.80 & 0.77 & 0.81 & 0.75 & 0.72\\ \hline
\bottomrule
\end{tabular}
\end{center}

\vspace{0.5cm}

\begin{center}
\captionof{table}{Ratio of the computing time of one iteration of Algorithm 1 and Algorithm 2.\label{runtab}}  
\begin{tabular}{c c}
\hline\hline
$p$ & Ratio of computing times\\
\midrule
40 & 4.93\\
60 & 4.82\\
100 & 4.81\\
\hline
\end{tabular}
\end{center}

\begin{center}
\captionof{table}{Change-point Estimation Results for different values of $p$ and different percentages of similarity for the Fast Implementation Algorithm.($T=700$, $s_1=s_2=\frac{10p(p+1)}{2}\%$, $\tau^*=354$)\label{estimresult}}  
\begin{tabular}{l c c c c c}  
\hline\hline                       
 p & \% of Similarity & $\widehat{\tau}$ & $\widehat{\widehat{\tau}}$ & RMSE & CV 
\\ [0.5ex]   
\hline
\multirow{3}{*}{40} & 0 & 360 & 360 & 17.89 & 0.04\\
                    & 20 & 363 & 361 & 30.07 & 0.08\\
										& 40 & 375 & 373 & 47.97 & 0.10\\ \hline
\multirow{3}{*}{60} & 0 & 357 & 356 & 23.05 & 0.06\\
                    & 20 & 388 & 386 & 43.20 & 0.08\\
										& 40 & 410 & 408 & 61.45 & 0.09\\ \hline
\multirow{3}{*}{100} & 0 & 356 & 355 & 35.93 & 0.10\\
                    & 20 & 408 & 401 & 62.89 & 0.10\\
										& 40 & 424 & 421 & 85.04 & 0.12\\ \hline
										
\end{tabular}
\end{center}

\begin{center}
\captionof{table}{Specificity, sensitivity and relative error of the two parameters for different values of $p$ and different percentages of similarity for the Fast Implementation Algorithm.\label{matresult}}  
\begin{tabular}{llllllll}
\hline\hline 
 p & \% of Similarity & \multicolumn{2}{c}{Specificity} & \multicolumn{2}{c}{Sensitivity} & \multicolumn{2}{c}{Relative error} \\
& & $\mathbf{\theta}_*^{(1)}$ & $\mathbf{\theta}_*^{(2)}$ & $\mathbf{\theta}_*^{(1)}$ & $\mathbf{\theta}_*^{(2)}$ & $\mathbf{\theta}_*^{(1)}$ & $\mathbf{\theta}_*^{(2)}$\\
    \midrule
    \multirow{3}{*}{40} & 0 & 0.74 & 0.86 & 0.78 & 0.86 & 0.74 & 0.67\\
                    & 20 & 0.74 & 0.81 & 0.76 & 0.82 & 0.73 & 0.71\\
										& 40 & 0.72 & 0.78 & 0.78 & 0.82 & 0.74 & 0.70\\ \hline
		\multirow{3}{*}{60} & 0 & 0.81 & 0.83 & 0.77 & 0.82 & 0.75 & 0.66\\
                    & 20 & 0.82 & 0.87 & 0.70 & 0.72 & 0.79 & 0.73\\
										& 40 & 0.80 & 0.86 & 0.65 & 0.68 & 0.81 & 0.78\\ \hline
		\multirow{3}{*}{100} & 0 & 0.82 & 0.88 & 0.75 & 0.84 & 0.78 & 0.66\\
                     & 20 & 0.81 & 0.87 & 0.66 & 0.70 & 0.81 & 0.78\\
										 & 40 & 0.85 & 0.87 & 0.63 & 0.68 & 0.83 & 0.81\\
    \bottomrule
\end{tabular}
\end{center}


\begin{figure}[ht]
\centering
\begin{tabular}{ccc}
\epsfig{file=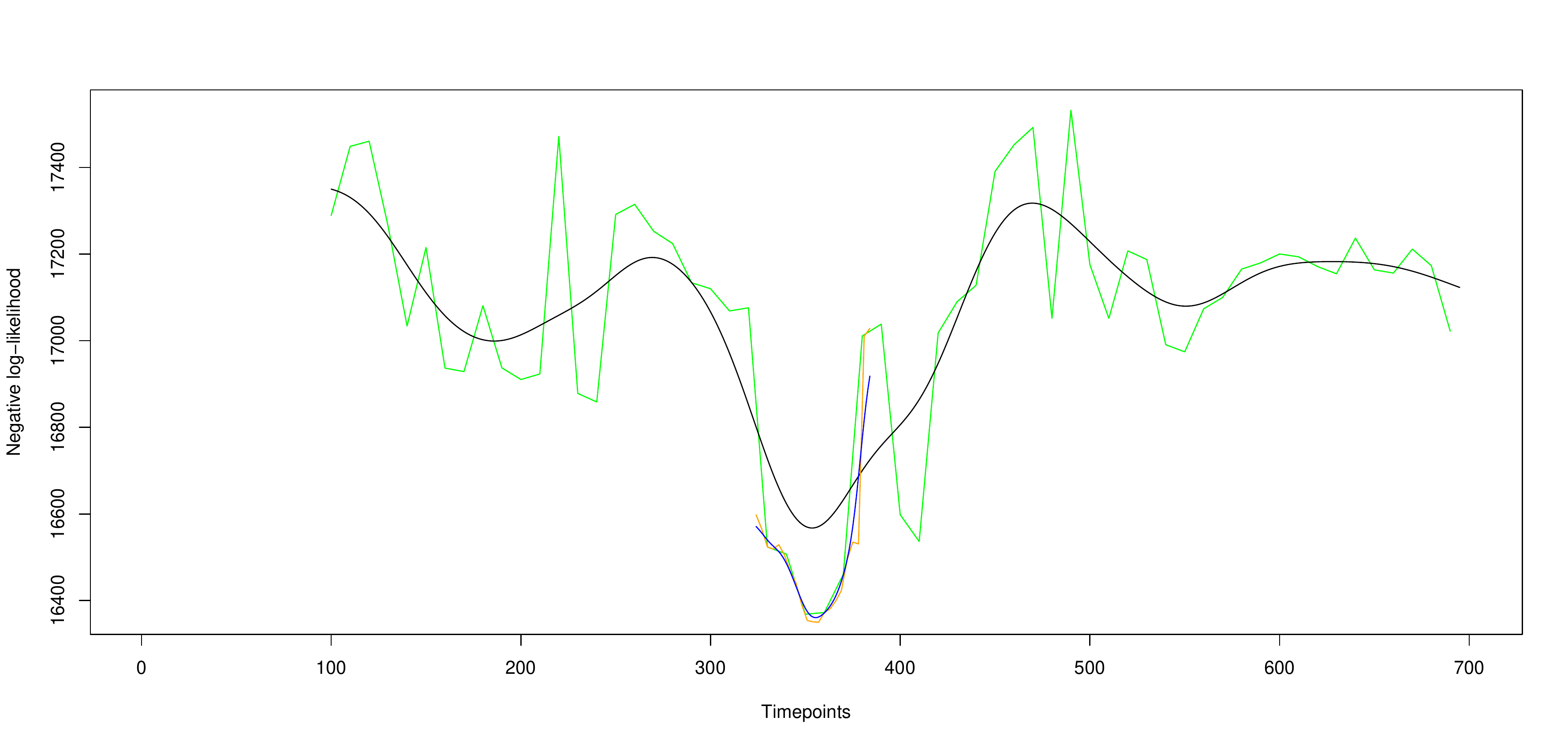,width=0.3\linewidth,clip=} & 
\epsfig{file=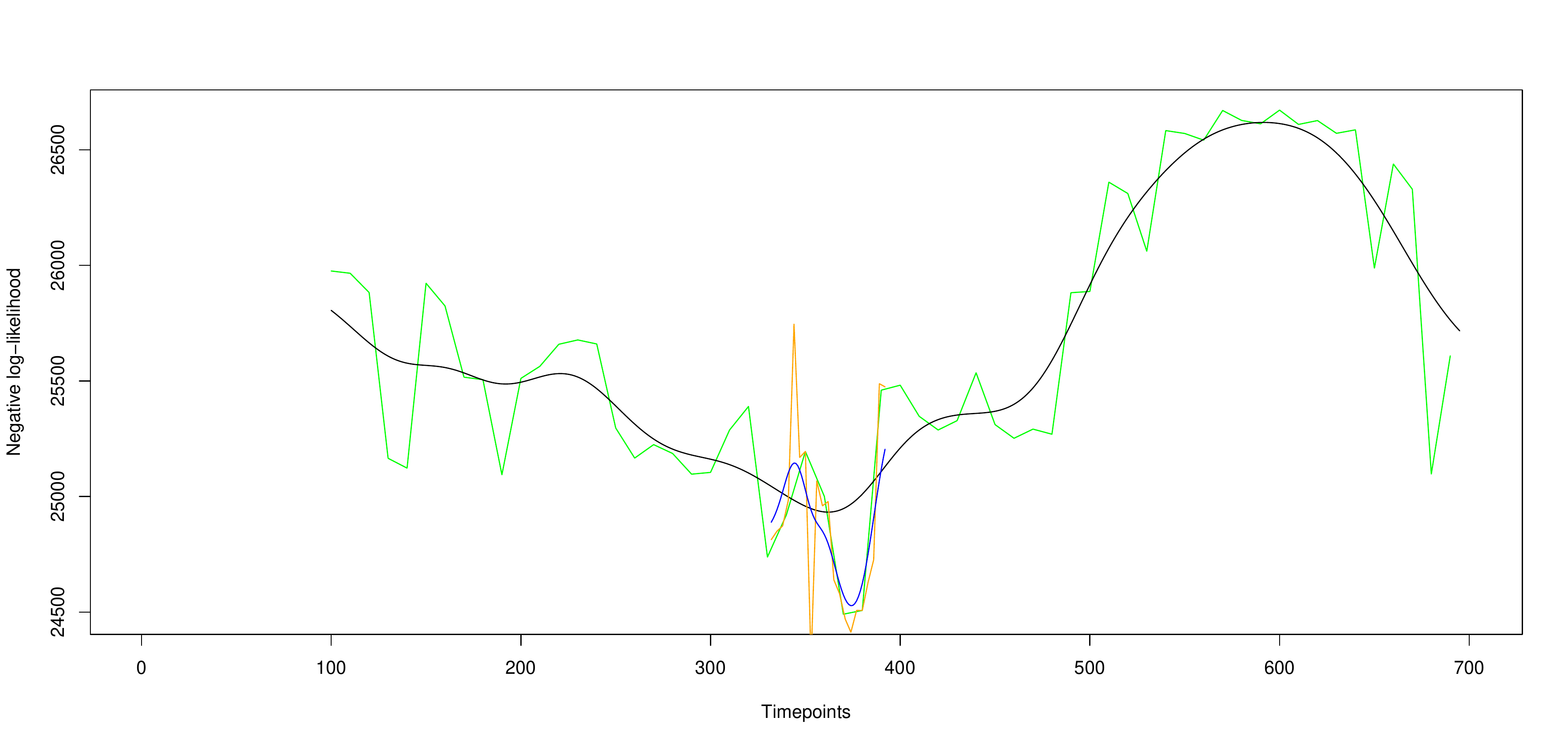,width=0.3\linewidth,clip=} &
\epsfig{file=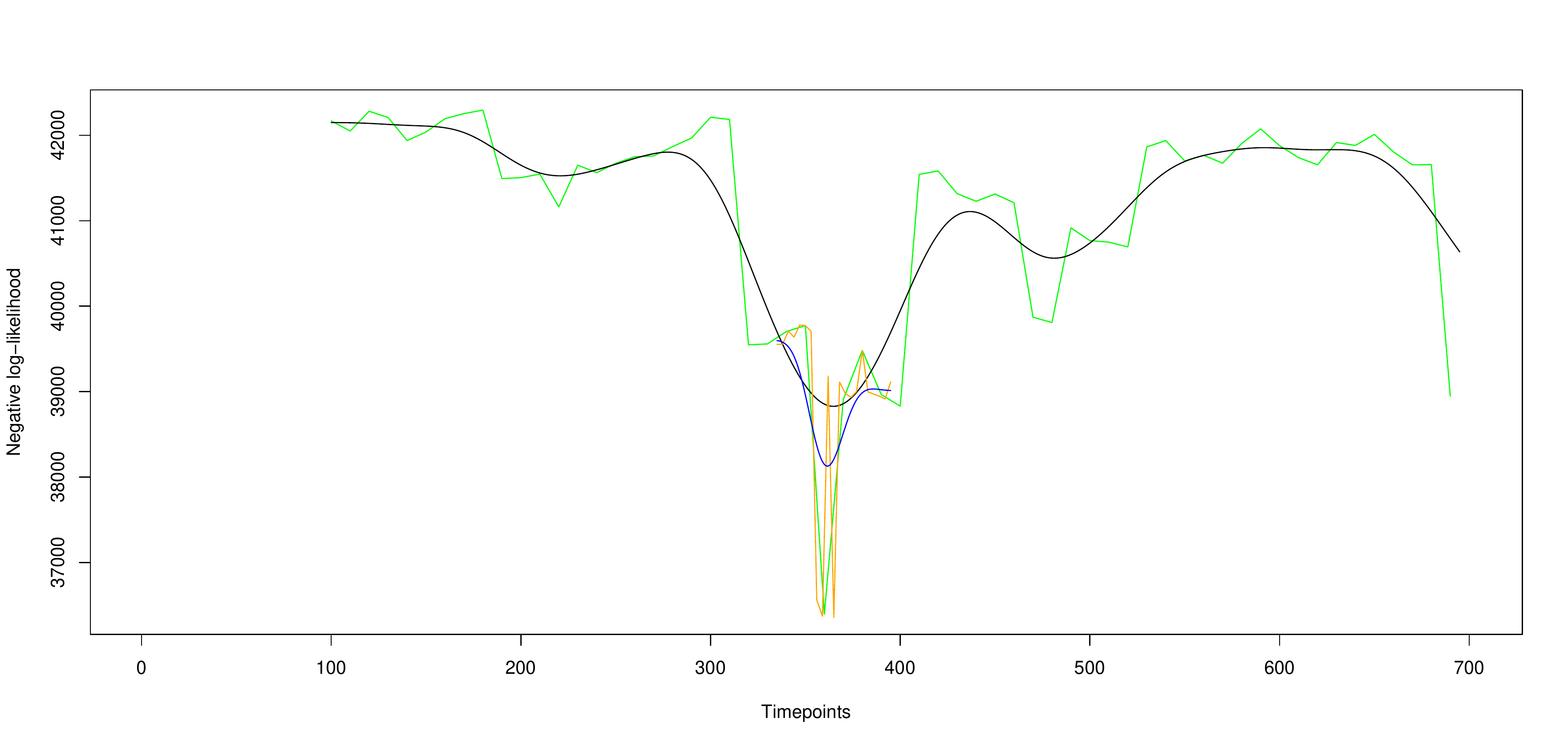,width=0.3\linewidth,clip=} \\
\epsfig{file=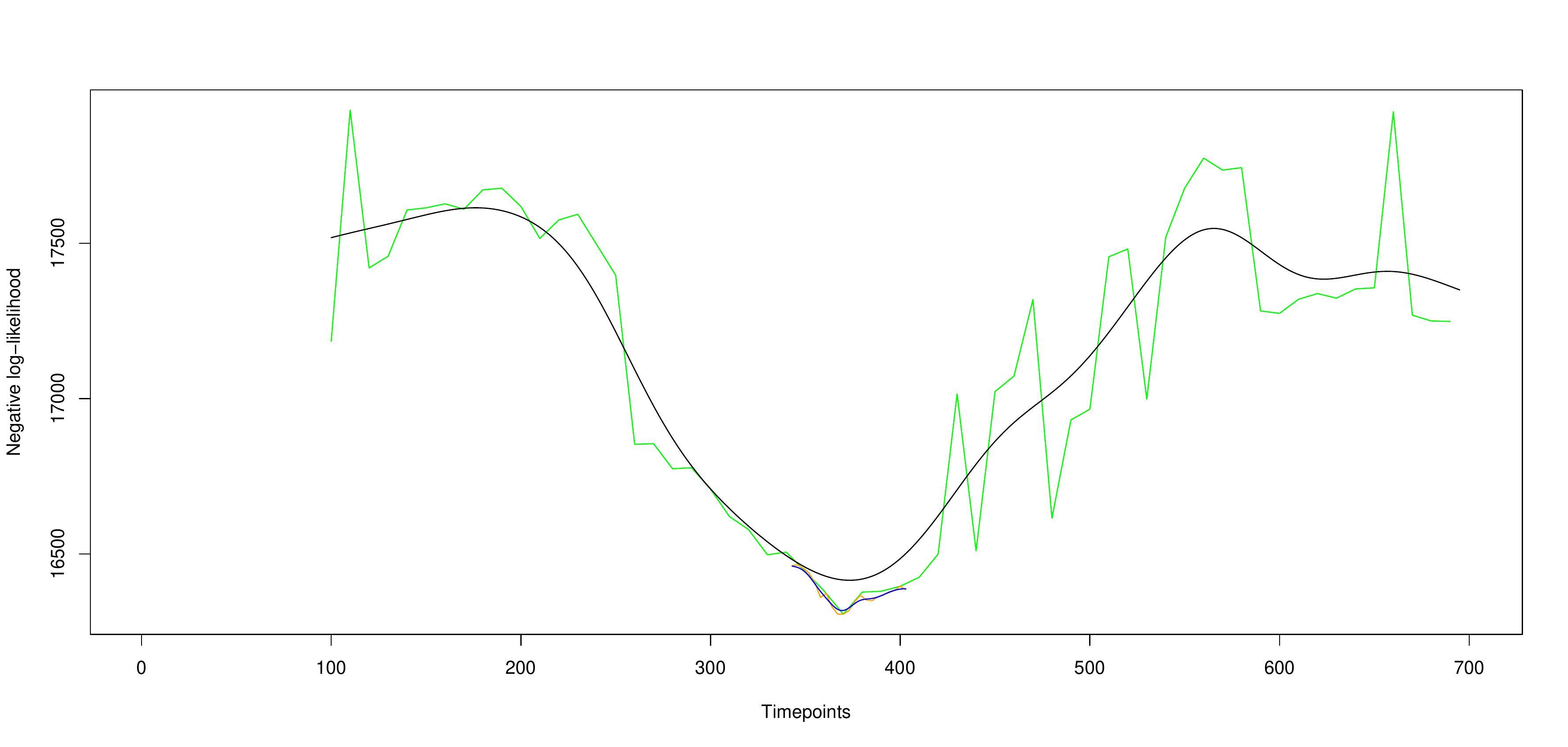,width=0.3\linewidth,clip=} &
\epsfig{file=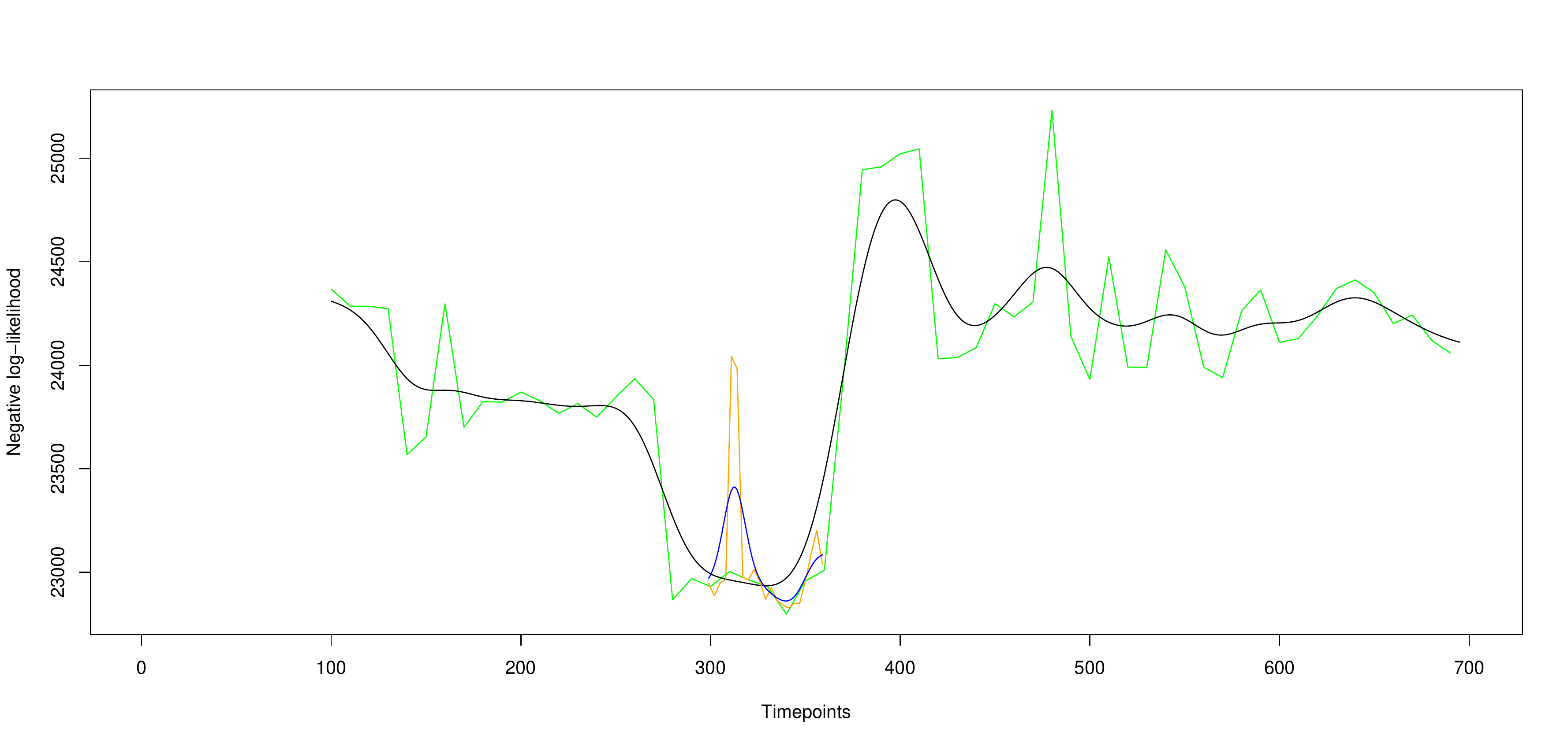,width=0.3\linewidth,clip=} &
\epsfig{file=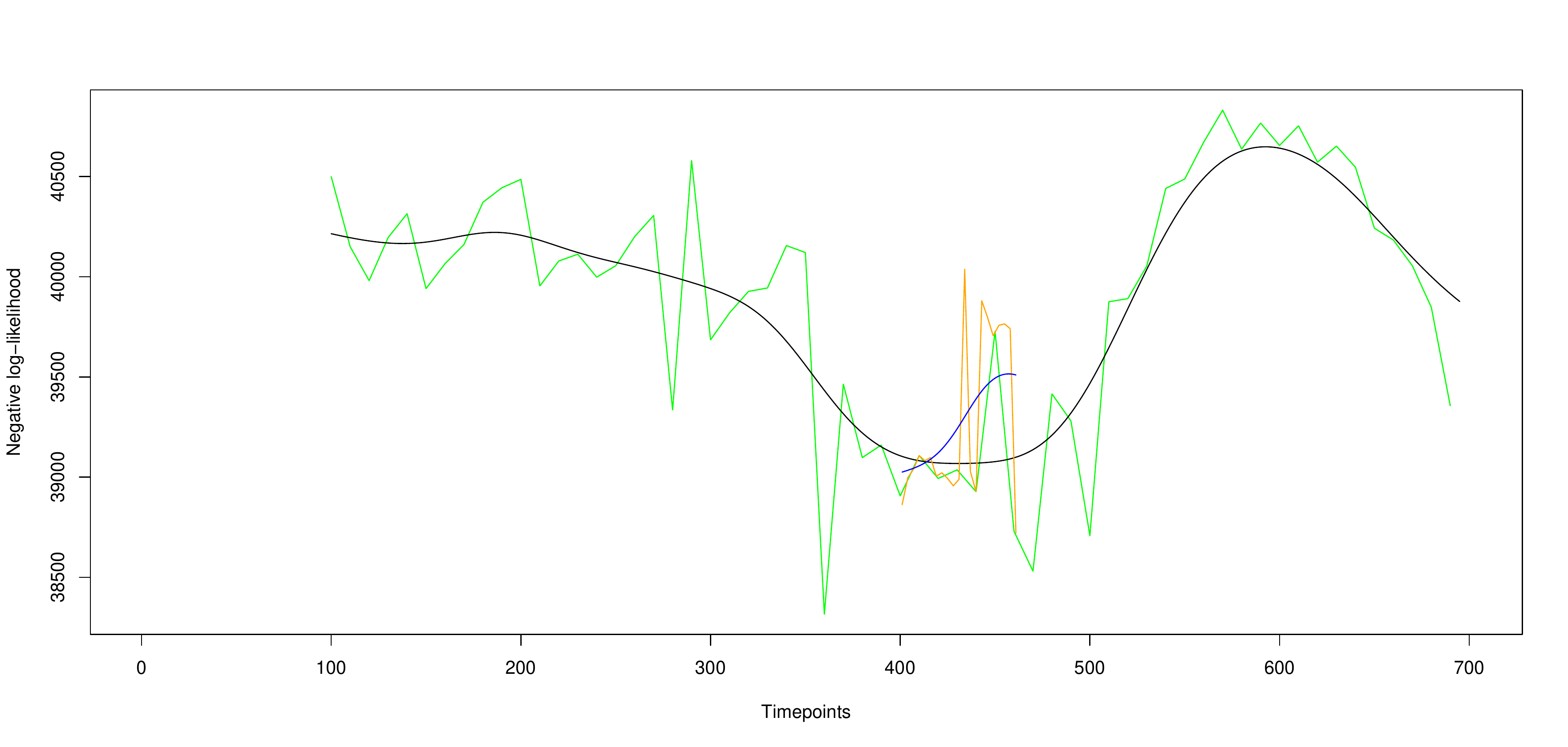,width=0.3\linewidth,clip=} \\
\epsfig{file=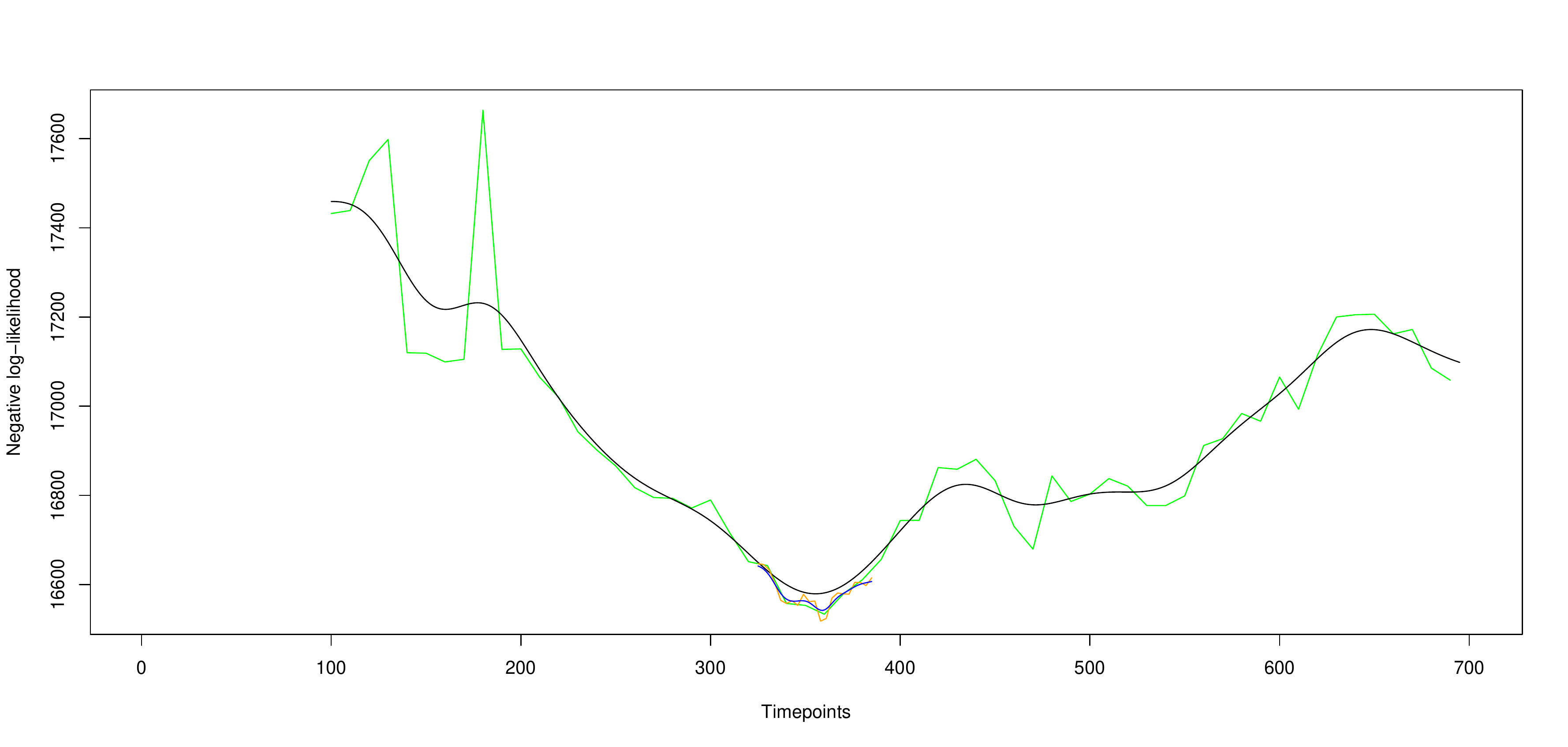,width=0.3\linewidth,clip=} &
\epsfig{file=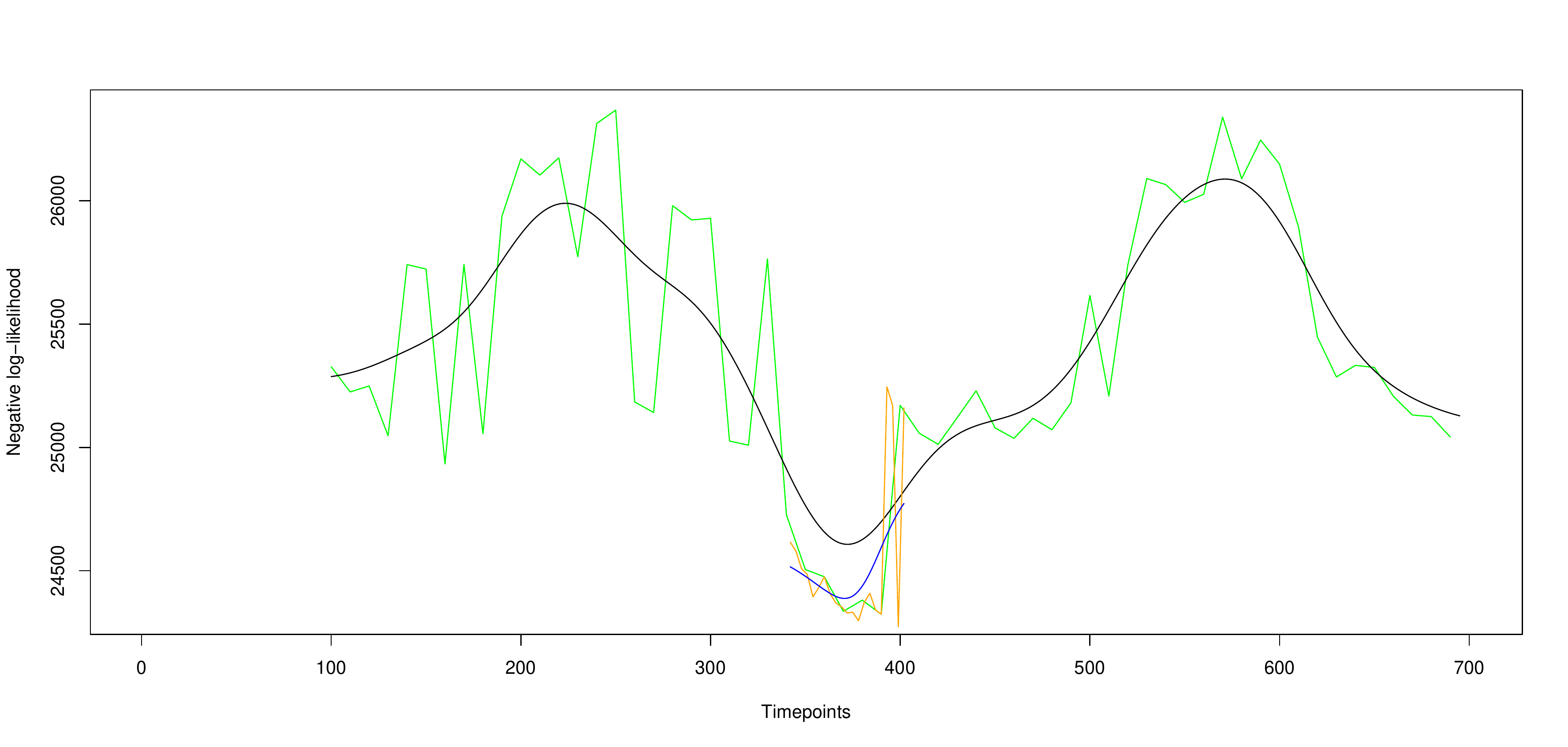,width=0.3\linewidth,clip=} &
\epsfig{file=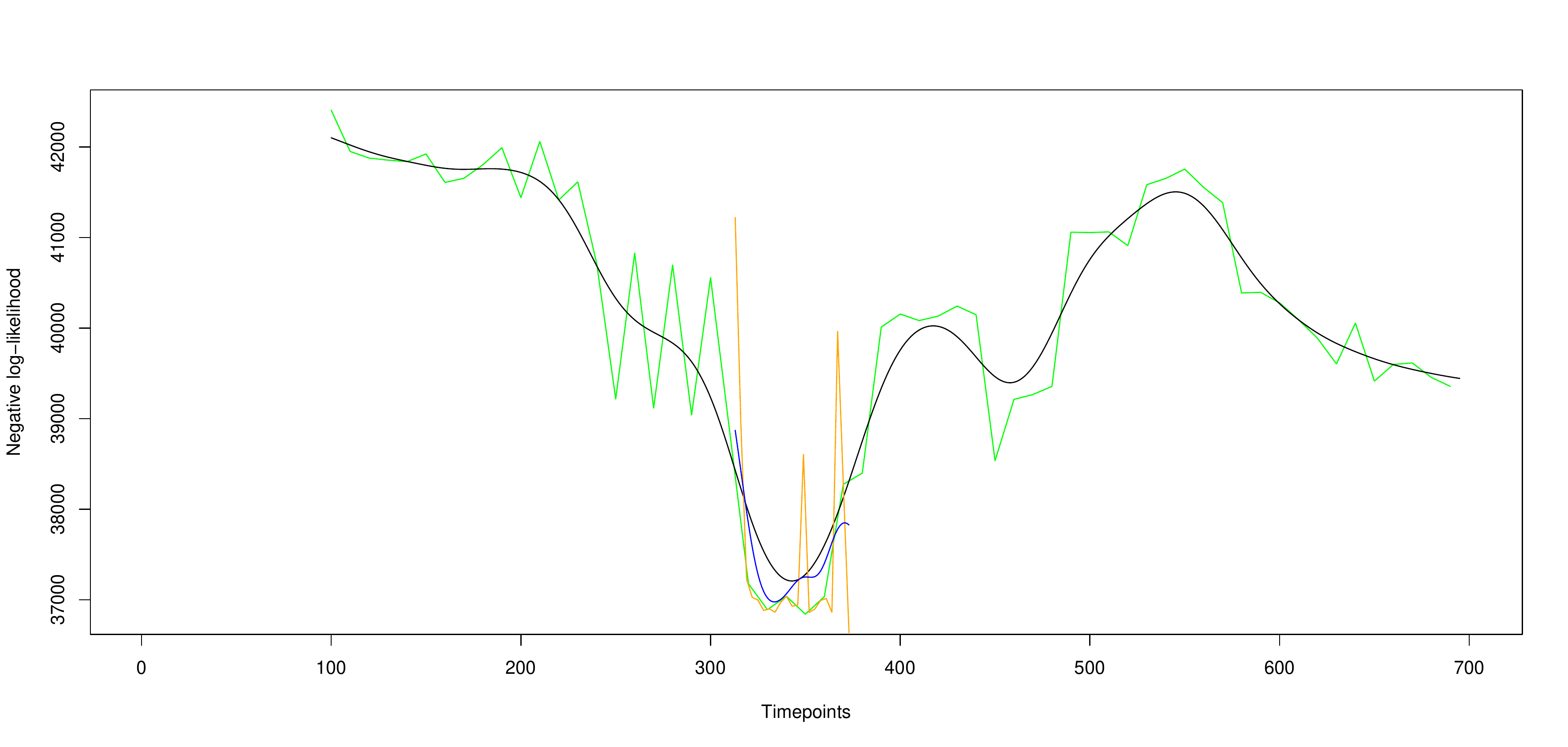,width=0.3\linewidth,clip=} \\
\end{tabular}
\caption{Smoothed profile pseudo-log-likelihood functions from one run of Algorithm 2. Different values of similarity ($0\%$, $20\%$ and $40\%$) in rows. Different values of $p$ ($p=40, 60\,\,\&\,\,100$) in column. The green curve is the non-smoothed profile pseudo-log-likelihood from Stage 1 of Algorithm 2, and the black curve is its smoothed version. The orange and the blue curve are respectively the non-smoothed and the smoothed profile pseudo-log-likelihood functions from Stage 2 of Algorithm 2.}
\label{fig1}
\end{figure}

\subsection{A community based network structure}
Next, we examine a setting similar to the one that emerges from the US Senate analysis presented in the next Section.
Specifically, there are two highly ``connected" communities of size $p=50$ that are more sparsely connected before the
change-point, but exhibit fairly strong negative association between their members after the change-point. Further,
the within community connections are increased for one of them and decreased for the other after the occurrence of the
change-point. We keep the density of the two matrices encoding the network structure before and after the true change-point at 10\%. In the pre change-point regime, 40\% of the non-zero entries are attributed to within group connections in community 1 (see Table ~\ref{table_community}), and 50\% to community 2 (see Table ~\ref{table_community}), while the remaining 10\% non-zeros represent between group connections and are negative. Note that the within group connections are all positive. In the post change-point regime, the community 1 within group connections slightly increase to 42\% of the non-zero entries, whereas those of community 2 decrease to 17\% of the non-zero entries. The between group connections increase to 41\% of the non-zero entries in the post change-point regime. As before, each off-diagonal element $\mathbf{\theta}_{jk}^{(i)}$, $i=1,2$ is drawn uniformly from $\left[-1,-0.5\right]\cup\left[0.5,1\right]$ if nodes $j$ and $k$ are linked by an edge, otherwise $\mathbf{\theta}_{*,jk}^{(i)}=0$, $i=1,2$ and the diagonals for both the matrices are assigned as zeros. Given the two matrices $\mathbf{\theta}_{*}^{(1)}$ and $\mathbf{\theta}_{*}^{(2)}$, we generate data using the ``BMN'' package (\cite{hofl}) as described earlier. The total sample size employed is $T=1500$ and the true change-point is at $\tau^*=750$. We choose the first stage grid comprising of 50 points with a step size of 27 and the second stage grid is chosen in a neighborhood of the first stage estimate with a step size of 3 with 20 points. We replicate the study 5 times and find that the estimated change-point averaged over the 5 replications as $\hat{\tau}=768$. The relevant figure (see Figure ~\ref{fig2}) for this two community model is given below. 
The analysis indicates that our proposed methodology is able to estimate the true change-point sufficiently well in the presence of varying degrees of connections between two communities over two different time periods, a reassuring feature for the US Senate application presented next.

\begin{center}
\captionof{table}{Positive and negative edges before and after the true change-point for two community model}
\begin{tabular}{|c|c|c|c|c|c|c|}
\hline
Edges	&	\multicolumn{3}{c|}{Before}	&	\multicolumn{3}{c|}{After}\\
\hhline{~------}
	&	comm 1	&	comm 2	& between & comm 1 & comm 2 &	between \\\hline
positive	& 50 & 63 & 0 & 52 & 21 & 0 \\	\hline
negative	& 0 & 0 & 10 & 0 & 0 & 50 \\	\hline
Total	& 50 & 63 & 10 & 52 & 21 & 50 \\	\hline
\hline
\end{tabular}
\label{table_community}
\end{center}

\begin{figure}[ht]
\centering
\includegraphics[scale=0.3]{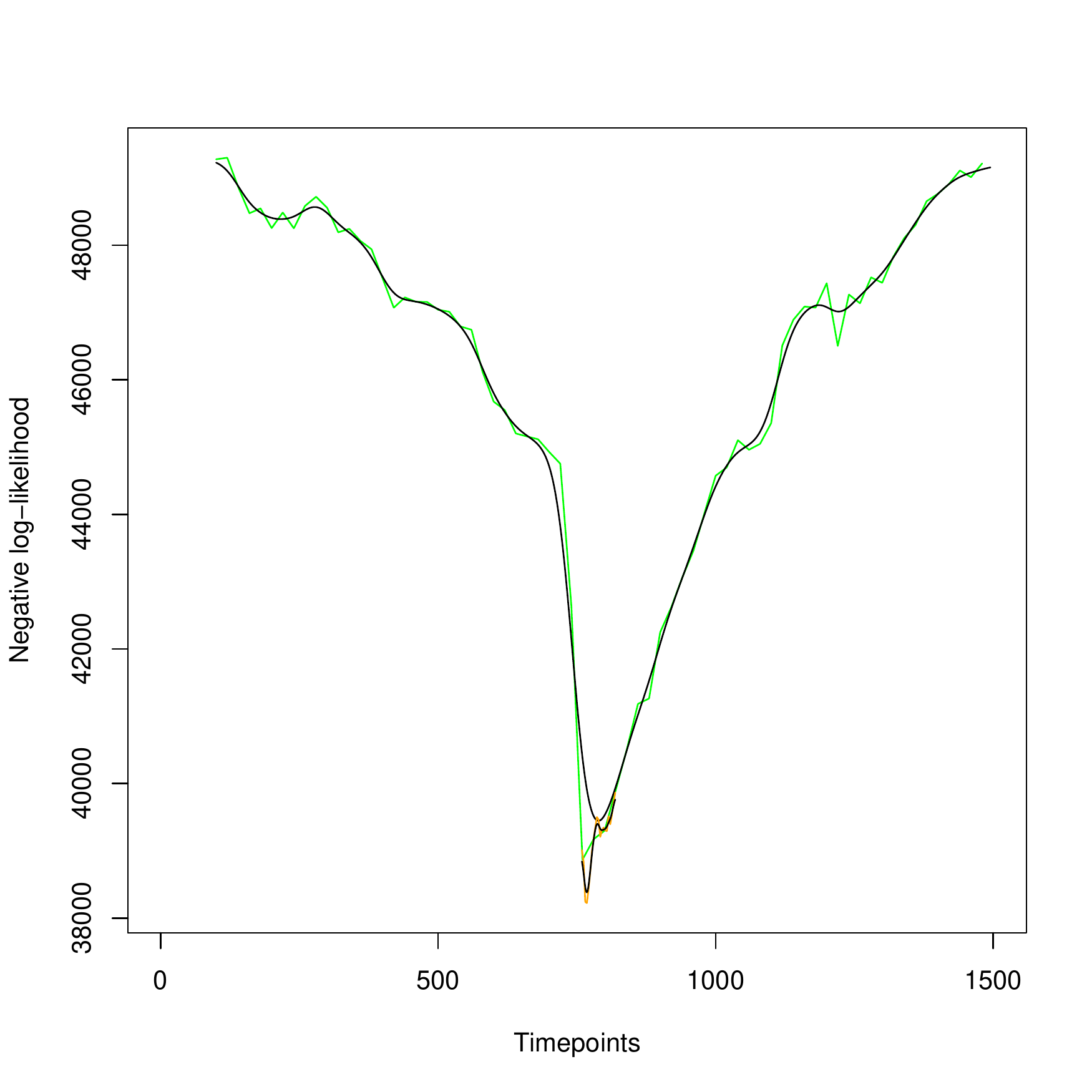}
\caption{Change-point estimate for the two community model with $p=50$, $T=1500$ and $\tau^*$=754}
\label{fig2}
\end{figure}

\section{Application to Roll Call Data of the US Senate}

The data examined correspond to voting records of the US Senate covering the period 1979 (96th Congress) to
2012 (112th Congress) and were obtained from the website {\tt www.voteview.com}. Specifically, for each of the 
12129 votes cast during this period, the following information is recorded: the date that the vote occurred 
and the response to the bill/resolution under consideration -yes/no, or abstain- of the 100 Senate members.
 Due to the length of the time period under consideration, there was significant turnover
of Senate members due to retirements, loss of re-election bids, appointments to cabinet or other administrative positions, or physical demise. In order to hold the number of nodes fixed to 100 (the membership size of
the US Senate at any point in time), we considered Senate seats (e.g. Michigan 1 and Michigan 2) and carefully mapped 
the senators to their corresponding seats, thus creating a continuous record of the voting pattern of each Senate seat.

Note that a significant number of the 12129 votes 
deal with fairly mundane procedural matters, thus resulting in nearly unanimous outcomes. Hence, only votes
exhibiting conformity less than 75\% (yes/no) in either direction were retained, thus resulting in an effective sample
size of $T=7949$ votes. Further, missing values due to abstentions were imputed by the value (yes/no) of that member's
party majority position on that particular vote. Note that other imputation methods of missing values were employed: (i) 
replacing all missing values by the value (yes/no) representing the winning majority on that bill and (ii) replacing the missing
value of a Senator by the value that the majority of the opposite party voted on that particular bill. The results based on these
two alternative imputation methods are given in the Supplement.

Finally, the yes/no votes were encoded as 1/0, respectively. Under
the posited model, votes are considered as i.i.d. from the same underlying distribution pre and post any change-point.
In reality, voting patterns are more complex and in all likelihood exhibit temporal dependence within the two year period
that a Congress serves and probably even beyond that due to the slow turnover of Senate members. Nevertheless, the
proposed model serves as a {\em working model} that captures essential features of the evolving voting dependency structure
between Senate seats over time.

The likelihood function together with an estimate of a change-point are depicted in Figure \ref{change-point} based on the 
Fast Implementation Algorithm presented in Section 4. We choose our first stage grid with a step-size of 50 that yields 157 points excluding time points close to both boundaries. In the second stage, we choose a finer-resolution grid with a step size of 20 in a neighborhood of the first stage change-point estimate. The vote corresponding to the change point occurred on January 17, 1995 at the beginning of the tenure of the 104th Congress. This change-point comes at the footsteps of the November 1994 election that witnessed the Republican Party capturing the US House of Representatives for the first time after 1956. As discussed in the political science literature, the 1994 election marked the end of the ``Conservative Coalition'', a bipartisan coalition of conservative oriented
Republicans and Democrats on President Roosevelt's ``New Deal'' policies, which had often managed to control Congressional outcomes since the ``New Deal'' era. Note that other analyses based on fairly ad hoc methods (e.g. \cite{moo}) also point to a significant change occurring after the November 1994 election.

Next, we examine more closely the pre and post change-point network structures, shown in the form of heatmaps of
the adjacency matrices in Figure~\ref{maps}. To obtain stable estimates of the respective network structures,
stability selection (\cite{mein2}) was employed with edges retained if they were present in more than 90\% of the 50 
networks estimated from bootstrapped data. To aid interpretation, the 100 Senate seats were assigned to three categories: Democrat (blue), mixed (yellow) and Republican (red). Specifically, a seat was assigned to the 
Democrat or Republican categories if it were held for more than 70\% of the time by the corresponding party within the 
pre or post change-point periods; otherwise, it was assigned to the mixed one. This means that if a seat was held for
more than 5 out of the 8 Congresses in the pre change-point period and similarly 6 out of 9 Congresses in the post period by
the Democrats, then it is assigned to that category and similarly for Republican assignments; otherwise, it is categorized as mixed.

In the depicted heatmaps, the ordering of the Senate seats in the pre and post change-point regimes are kept as similar 
as possible, since some of the seats changed their category membership completely across periods. Further, the green dots represent positive edge weights, mostly corresponding to within categories interactions, while black dots represent negative edge weights, mostly between category interactions. It can be clearly seen an emergence of a significant number of black dots in the post change-point regimes, indicative of sharper disagreements between political parties and thus increased polarization. Further, it can be seen that in the post change-point regime the mixed group becomes more prominent, indicating that it contributes to the emergence of a change-point.

To further explore the reasons behind the presence of a change-point, we provide some network statistics in Figure~\ref{fig3} and Figure~\ref{fig4}. Specifically, the two figures present the proportion of positive and negative edges, before and after the estimated change-point using two different methods for selecting the penalty tuning parameters; an analogue of the Bayesian Information Criterion and threshold 0.8 for the stability selection  method respectively.
The patterns shown across the figures for the two different methods are very similar- high proportion of positive edges within groups and very low or almost negligible proportion of negative edges within the ``{\color{red}republican}'' or  ``{\color{blue}democrat}'' groups in both pre and post-change-point periods. Further, a large proportion of negative edges can be accounted for ``{\color{red}republican}'' and ``{\color{blue}democrat}''  group interactions, which tend to increase in the post regime.
 One noticeable fact is that the proportion of positive edges within the ``{\color{red}republican}'' and ``{\color{blue}democrat}'' groups remain almost same from pre to post change-point regime under BIC and stability selection both whereas the proportion of positive edges between the two groups decrease and the proportion of negative edges between them tend to increase from pre to post change-point regime for both the methods. It can also be observed that the  ``{\color{yellow}mixed}'' and the ``{\color{blue}democrat}''  groups exhibit a large proportion of positive edges between them in the pre regime, as gleaned from their overlap in the corresponding heatmap. 

We also present some other network statistics, such as average degree, centrality scores and average clustering coefficients for the three groups ``{\color{red}republican}'', ``{\color{blue}democrat}'' and ``{\color{yellow}mixed}'' in Table~\ref{net_stat1}. We observe that in terms of centrality scores  
the ``{\color{blue}democrat}'' group is more influential than the ``{\color{red}republican}'' one, in both the pre and post change-point network structures,
 whereas in terms of clustering coefficient values the ``{\color{red}republican}'' group is ahead of the ``{\color{blue}democrat}'' one and the gap increases from pre to post change-point regime, also reflected in the finding that the number of edges within the ``{\color{red}republican}'' group mostly remains the same from pre to post regimes, whereas for the democrats it decreases. These results suggest that the Republicans form a tight cluster, whereas the Democrats not to
 the same extent.


\begin{figure}[ht]
\centering
\begin{tabular}{cc}
\epsfig{file=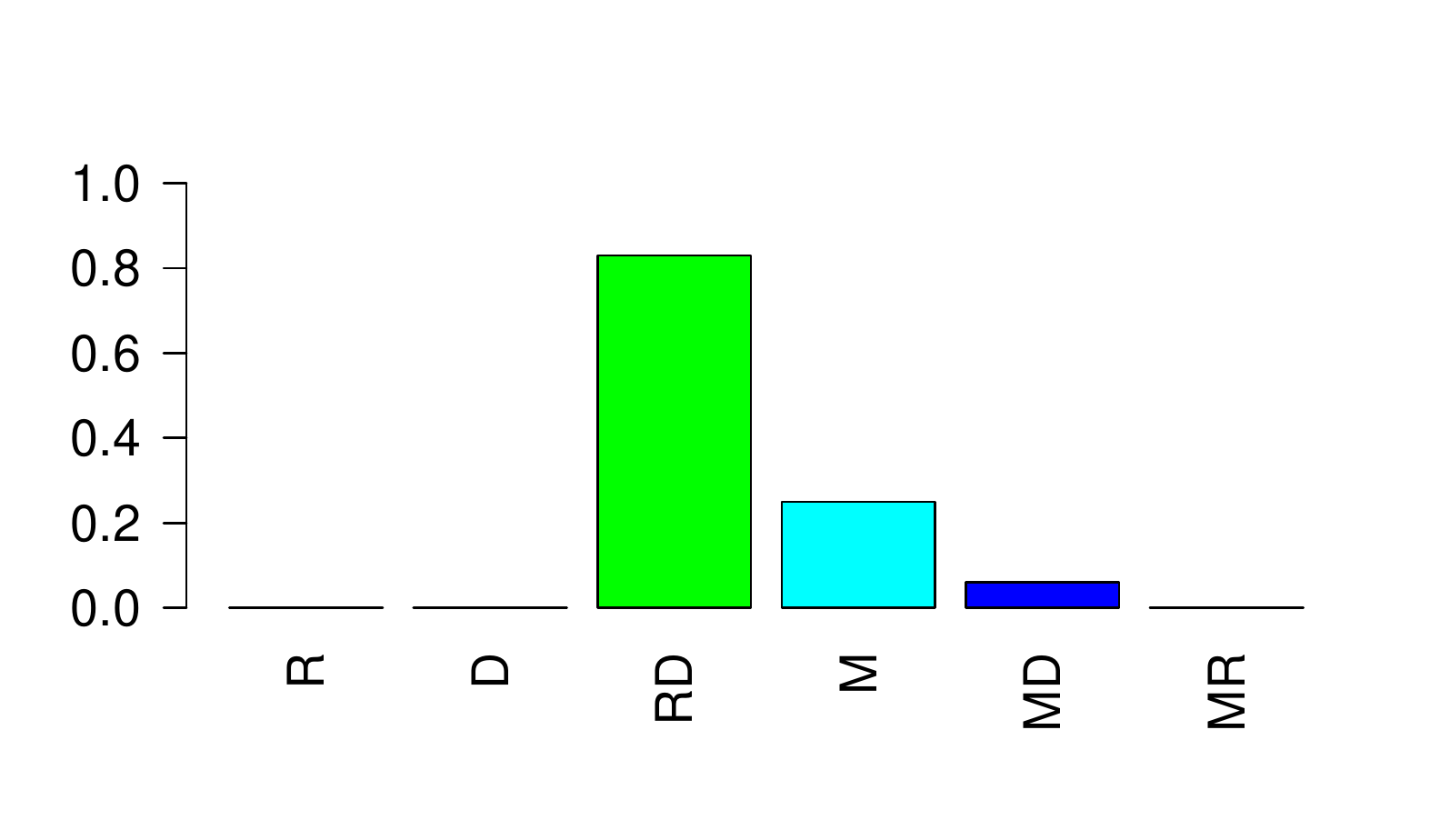,width=0.3\linewidth,clip=} & 
\epsfig{file=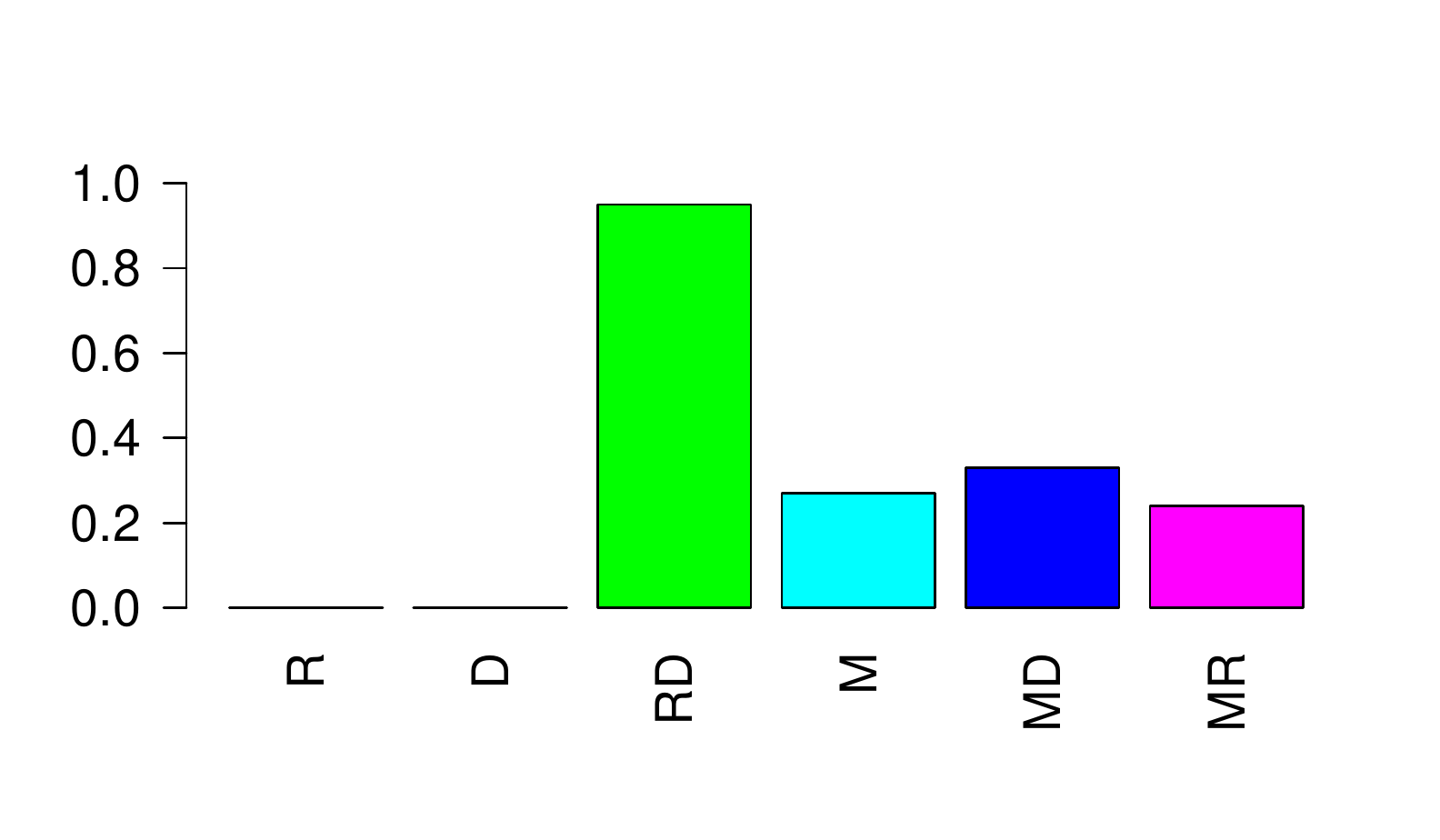,width=0.3\linewidth,clip=} \\
\epsfig{file=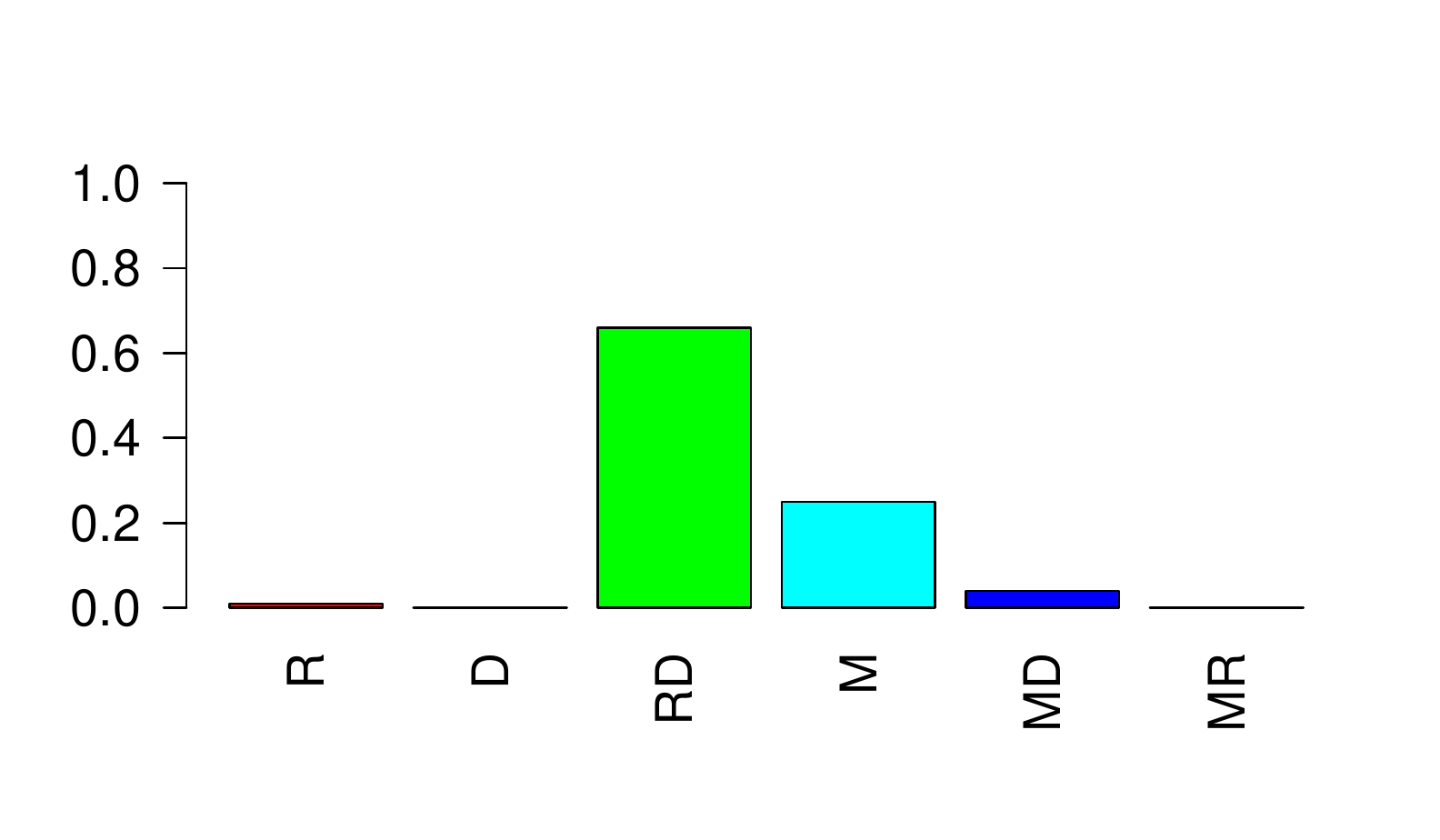,width=0.3\linewidth,clip=} &
\epsfig{file=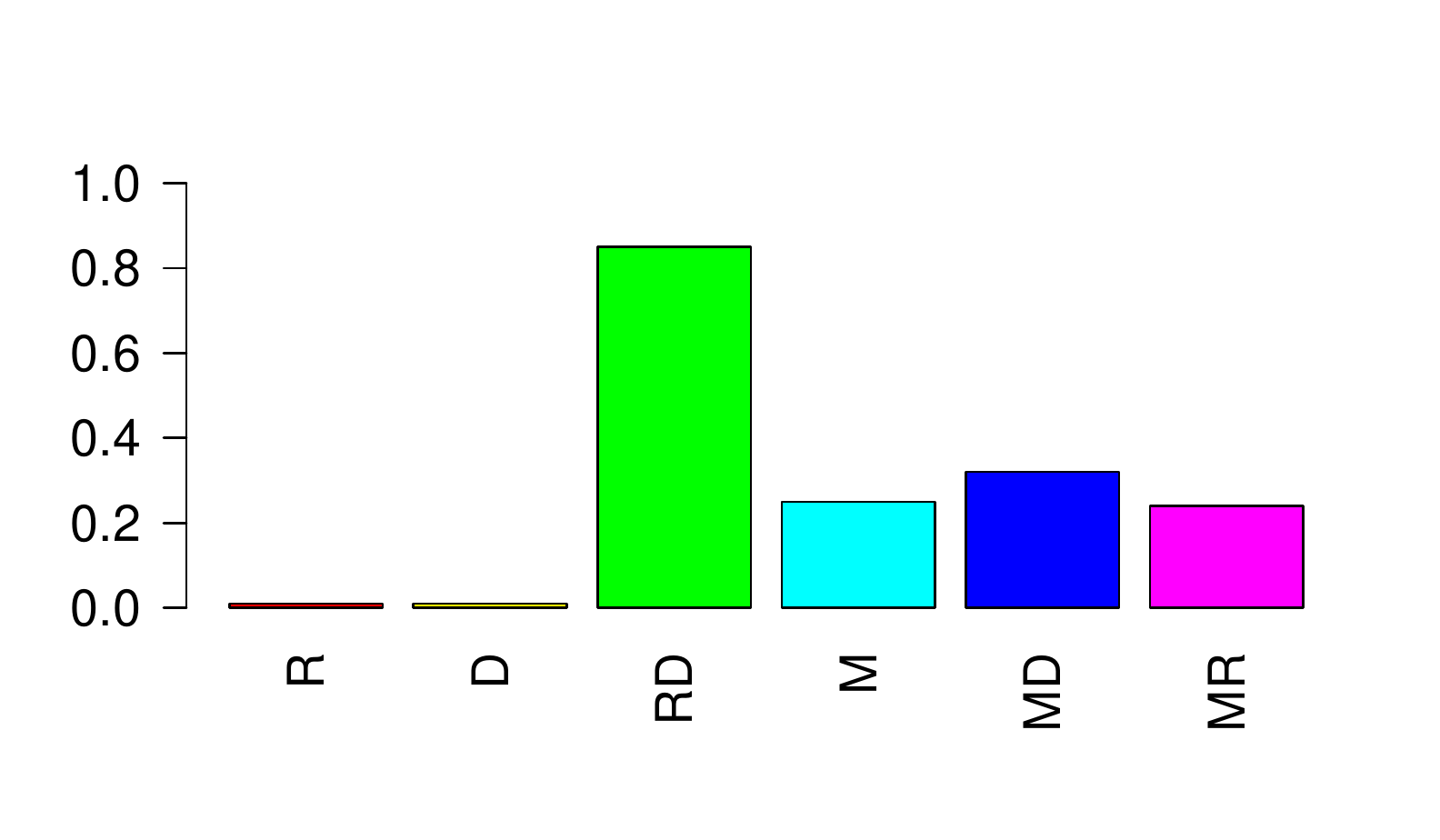,width=0.3\linewidth,clip=} \\
\end{tabular}
\caption{Proportion of negative edges for network structures before (left figure) and after (right figure) the estimated change-point for BIC and stability selection with threshold=0.8}
\label{fig3}
\end{figure}

\begin{figure}[ht]
\centering
\begin{tabular}{cc}
\epsfig{file=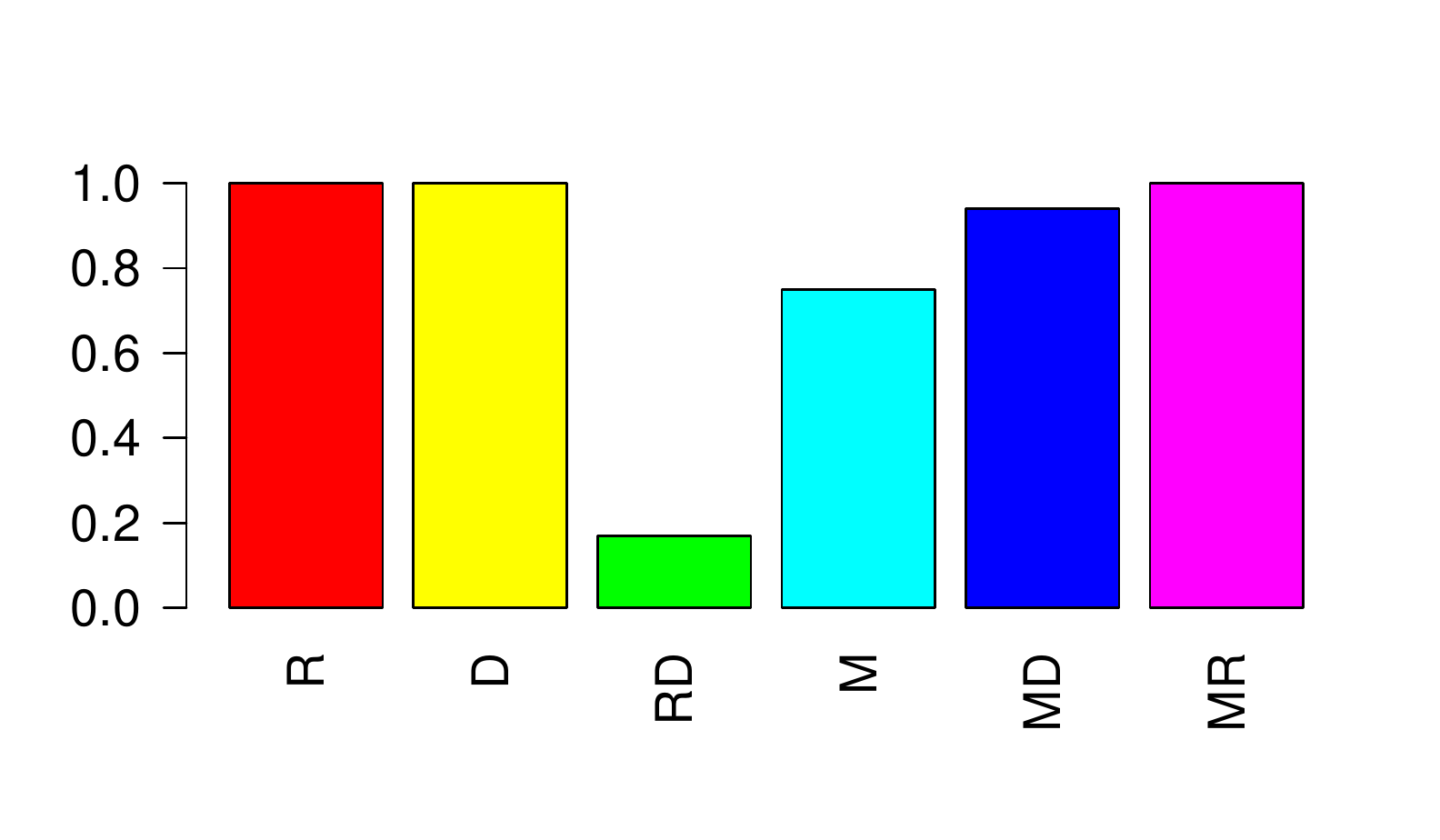,width=0.3\linewidth,clip=} & 
\epsfig{file=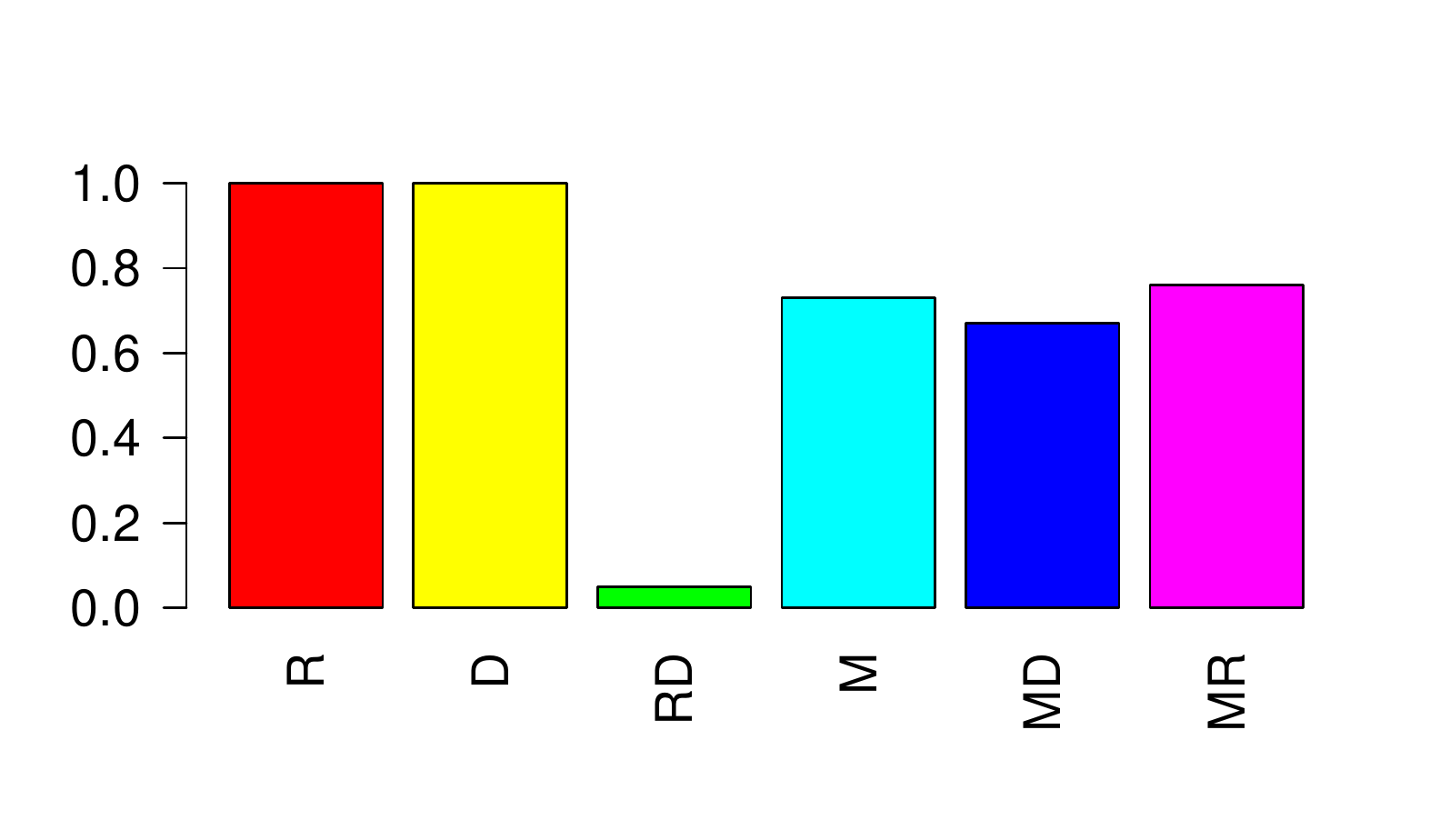,width=0.3\linewidth,clip=} \\
\epsfig{file=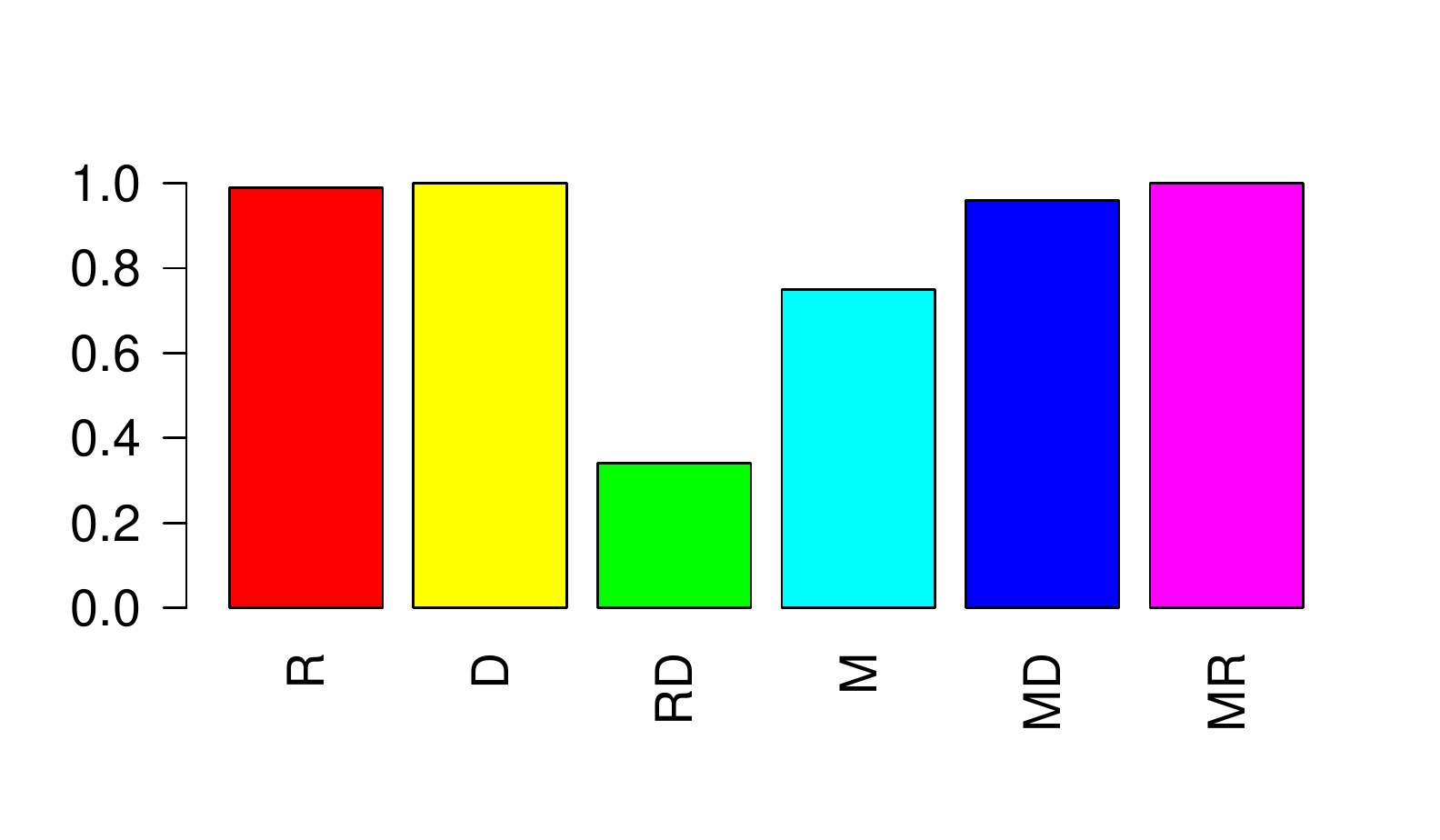,width=0.3\linewidth,clip=} &
\epsfig{file=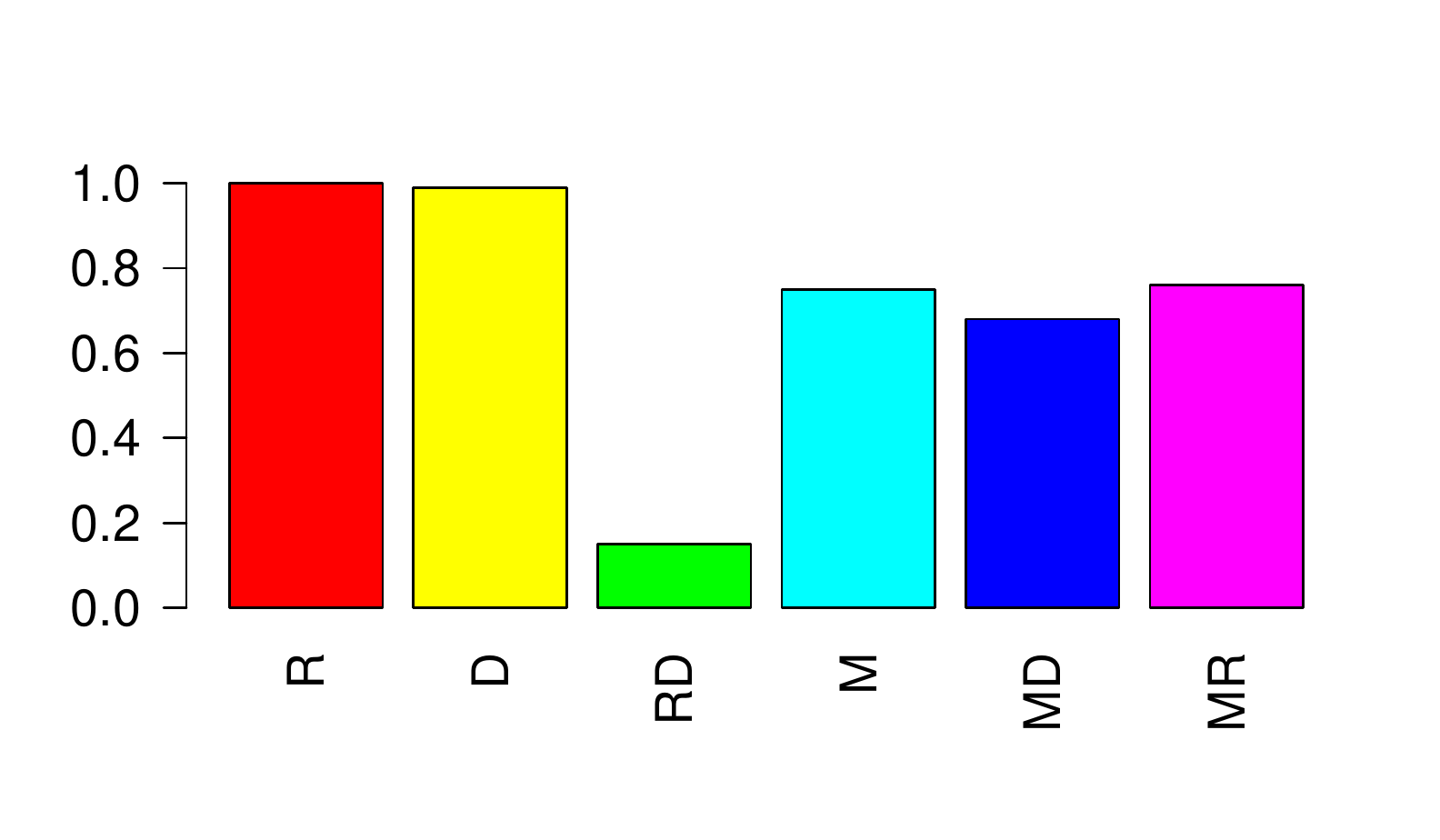,width=0.3\linewidth,clip=} \\
\end{tabular}
\caption{Proportion of positive edges for network structures before (left figure) and after (right figure) the estimated change-point for BIC and stability selection with threshold=0.8}
\label{fig4}
\end{figure}

\begin{center}
\captionof{table}{Different network statistic values for stability selection with threshold=0.9 and 0.8 respectively}
\label{net_stat1}

\begin{tabular}{|c|c|c|c|c|c|c|c|}
\hline
Methods & Network Statistic	&	\multicolumn{3}{c|}{Before}	&	\multicolumn{3}{c|}{After}\\
\hhline{~~------}
	& &	Rep	&	Dem	&	Mixed	&	Rep	&	Dem	&	Mixed\\\hline
Stable (0.9) & Centrality Score & 0.004 & 0.368 & 0.054 & 0.001 & 0.483 & 0.034 \\	\hline
& Clustering Coefficient & 0.346 & 0.311 & 0.339 & 0.334 & 0.251 & 0.391 \\	\hline
\hline
& & & & & & & \\\hline
Stable (0.8) & Centrality Score	& 0.004 & 0.378 & 0.055 & 0.001 & 0.481 & 0.078 \\	\hline
& Clustering Coefficient & 0.366 & 0.371 & 0.360 & 0.378 & 0.307 & 0.364 \\	\hline
\hline
\end{tabular}

\end{center}

\begin{figure}[ht]
\centering
\includegraphics[scale=0.3]{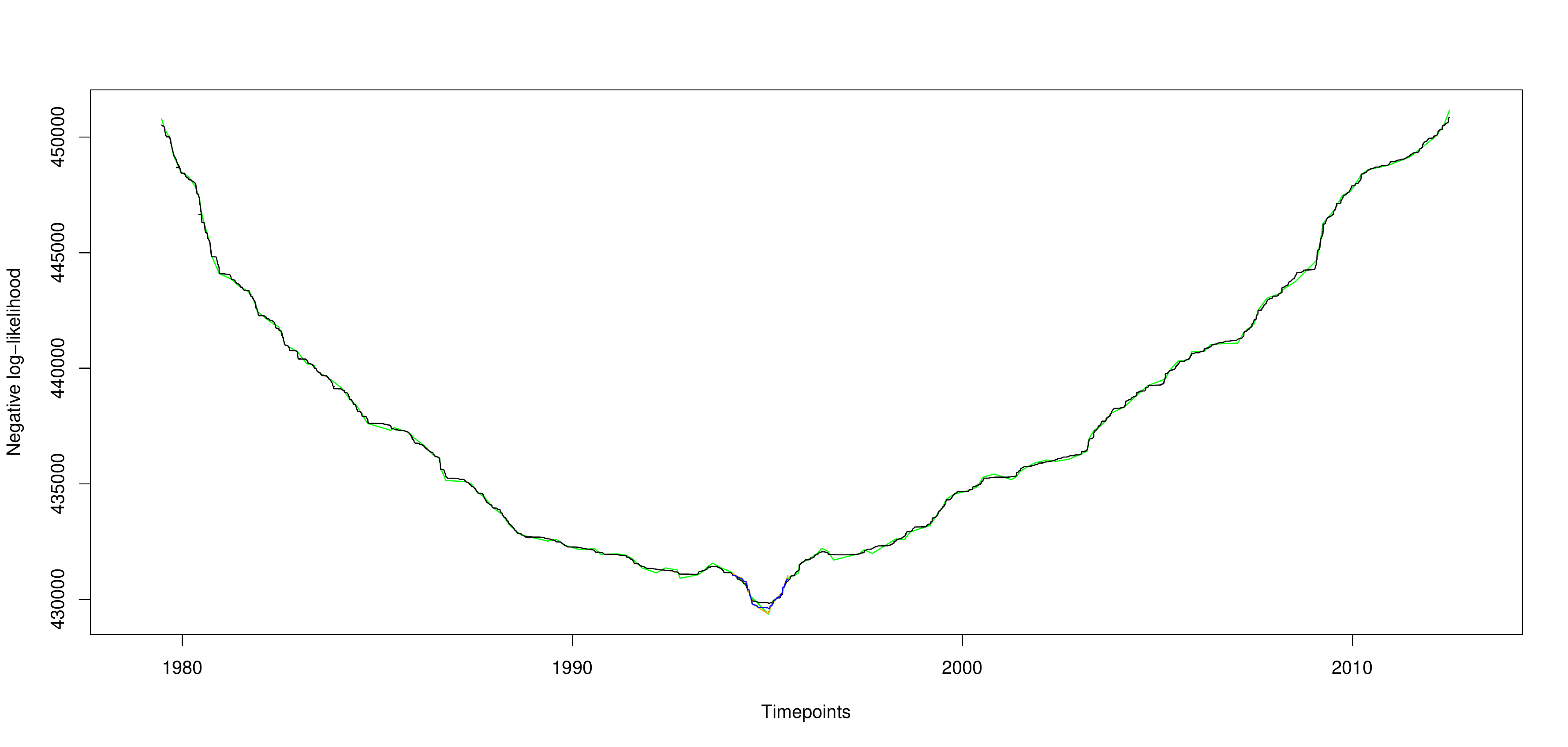}
\caption{Estimate of the change-point for the combined US senate data from 1979-2012}
\label{change-point}
\end{figure}

\begin{figure}[ht]
\begin{tabular}{cc}
\hspace{-70pt}\epsfig{file=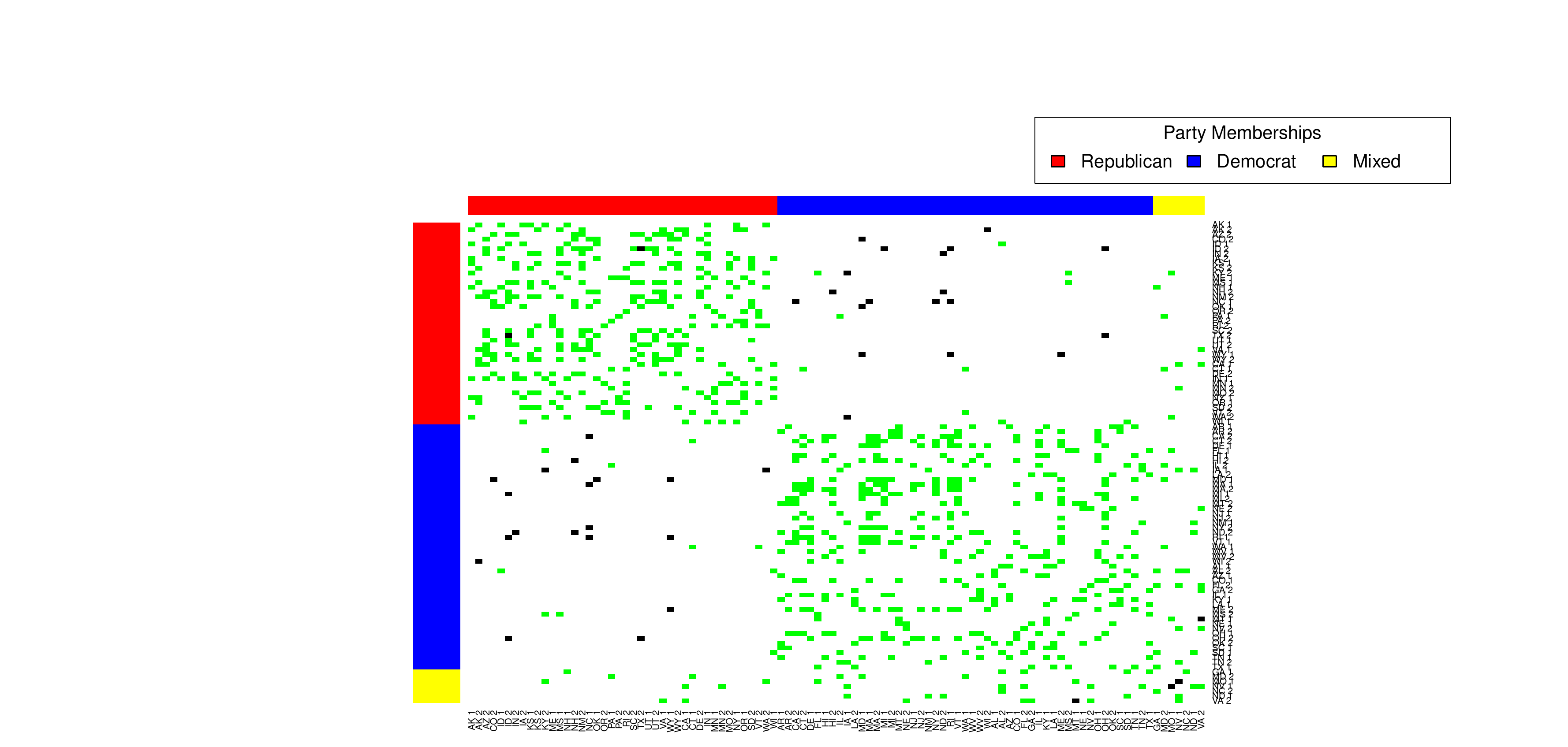,width=0.7\linewidth,clip=} & 
\hspace{-50pt}\epsfig{file=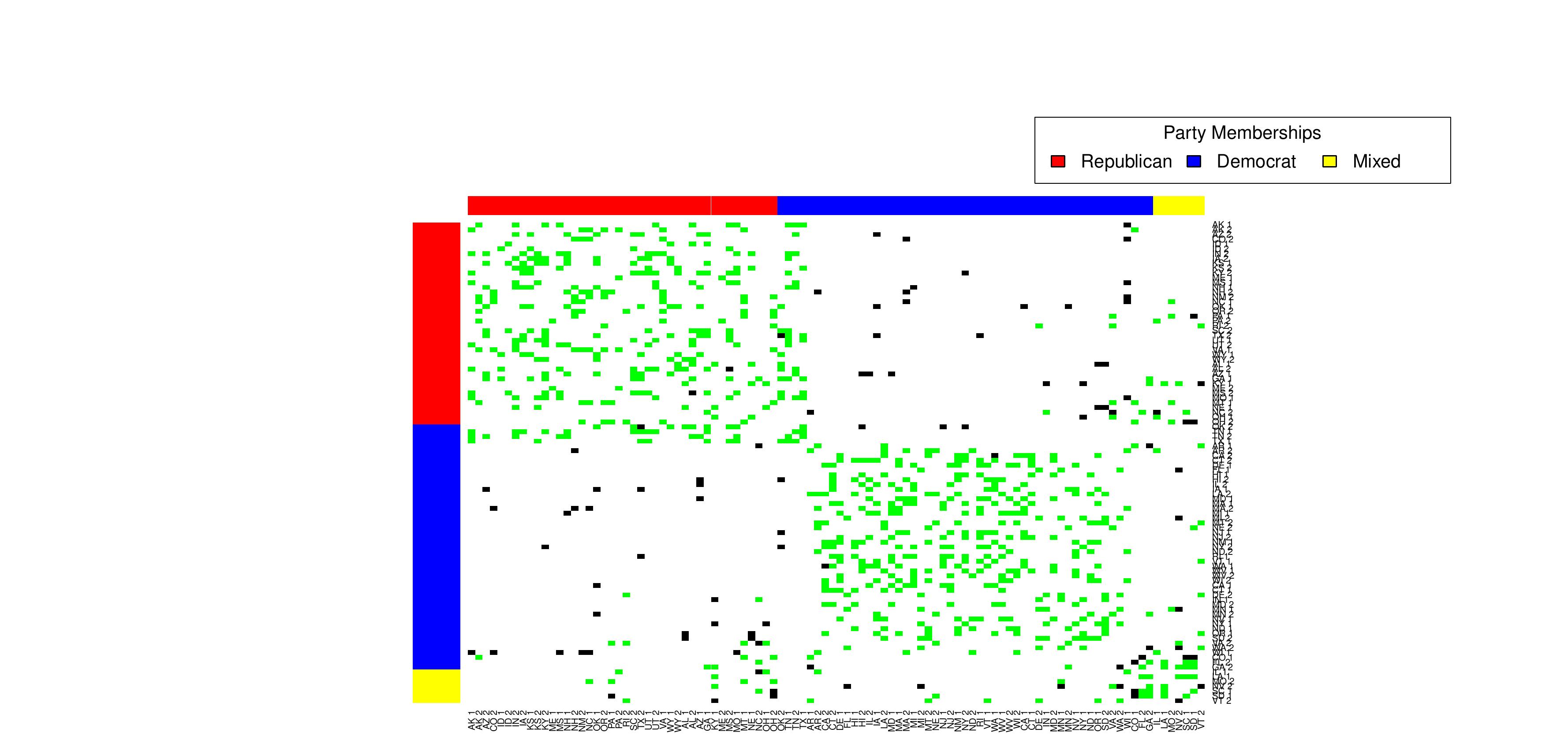,width=0.7\linewidth,clip=}\\
\end{tabular}
\caption{Heatmap of the stable network structures before and after the estimated change-point}
\label{maps}
\end{figure}



\renewcommand{\thepage}{S\arabic{page}}  
\renewcommand{\thesection}{S\arabic{section}}   
\renewcommand{\thetable}{S\arabic{table}}   
\renewcommand{\thefigure}{S\arabic{figure}}
\renewcommand{\theequation}{S\arabic{equation}}
\renewcommand{\thetheoremA}{S\arabic{theoremA}}

\vspace{10mm}\textbf{\Large Supplementary Information}\\\\Although our main motivation is in discrete graphical models, the proposed methodology can be applied more broadly for model-based change-point estimation. With this in mind, we shall prove a more general result that can be useful with other high-dimensional change-point estimation problems. Theorem \ref{thm1} follows as a special case.
\section{High-dimensional model-based change-point detection}
Let $\{X^{(t)},\;1\leq t\leq T\}$ be a sequence of $\rset^p$-valued independent random variables. Let $\Theta\subseteq\rset^d$ be an open, non-empty convex parameter space equipped with the Euclidean inner product $\seq{\cdot,\cdot}$, and norm$\|\cdot\|_2$. We will also use the $\ell^1$-norm $\|\theta\|_1\eqdef\sum_{j=1}^d |\theta_j|$, and the $\ell^\infty$-norm $\|\theta\|_\infty \eqdef\max_{1\leq j\leq d}|\theta_j|$. We assume that there exists a change point $\tau_\star\in \{1,\ldots,T-1\}$, parameters $\theta_\star^{(1)},\theta_\star^{(2)}\in\Theta$, such that for $t=1,\ldots,\tau_\star$, $X^{(t)}\sim g^{(t)}_{\theta^{(1)}_\star}$, and for $t=\tau_\star +1,\ldots,T$, $X^{(t)}\sim g^{(t)}_{\theta^{(2)}_\star}$, where $g^{(t)}_{\theta^{(1)}_\star}$ and $g^{(t)}_{\theta^{(2)}_\star}$ are probability densities on $\rset^p$. The goal is to estimate $\tau_\star,\theta_\star^{(1)},\theta_\star^{(2)}$. This setting includes the Markov random field setting (our main motivation), where $g^{(t)}_{\theta^{(1)}_\star}$ and $g^{(t)}_{\theta^{(2)}_\star}$ does not depend $t$. It also includes regression models where the index $t$ in the distributions $g^{(t)}_{\theta^{(1)}_\star}$ and $g^{(t)}_{\theta^{(2)}_\star}$ accounts for the covariates of subject $t$.

For $t=1\ldots,T$, let $(\theta,x)\mapsto \phi_t(\theta,x)$ be jointly measurable functions on $\Theta\times \rset^p$, such that $\theta\mapsto \phi_t(\theta,x)$ is convex and continuously differentiable for all $x\in\rset^p$. We  define
\[\ell_T(\tau;\theta_1,\theta_2) \eqdef \frac{1}{T}\sum_{t=1}^\tau \phi_t(\theta_1,X^{(t)}) + \frac{1}{T}\sum_{t=\tau+1}^T \phi_t(\theta_2,X^{(t)}),\]
and we consider  the change-point estimator $\tau_\star$ given by
\begin{equation}\label{c_p_estimator_suppl}
\widehat{\tau}=\argmin_{\tau\in\mathcal{T}}\,\ell_T(\tau;\widehat{\mathbf{\theta}}_{1,\tau},\widehat{\mathbf{\theta}}_{2,\tau}),
\end{equation}
for a non-empty search domain $\mathcal{T}\subset\{1,\ldots,T\}$, where for each $\tau\in\mathcal{T}$, $\widehat{\mathbf{\theta}}_{1,\tau}$ and $\widehat{\mathbf{\theta}}_{2,\tau}$ are defined as
\[\widehat{\mathbf{\theta}}_{1,\tau}\eqdef \argmin_{\theta\in\Theta}\left[\frac{1}{T}\sum_{t=1}^\tau \phi_t(\theta,X^{(t)}) + \lambda_{1,\tau}\|\mathbf{\theta}\|_{1}\right],\]
 and 
\[\widehat{\mathbf{\theta}}_{2,\tau}\eqdef \argmin_{\theta\in\Theta}\left[\frac{1}{T}\sum_{t=\tau+1}^T \phi_t(\theta,X^{(t)}) + \lambda_{2,\tau}\|\mathbf{\theta}\|_{1}\right],\]
for some positive penalty parameters $\lambda_{1,\tau},\lambda_{2,\tau}$. Note that by allowing the use of user-defined learning functions $\phi_t$, our framework can be used to analyze maximum likelihood and maximum pseudo-likelihood change-point estimators.

For $\tau\in\{1,\ldots,T-1\}$, we set
\[
\mathcal{G}_\tau^1 \eqdef \frac{1}{T}\sum_{t=1}^{\tau}\nabla \phi_{t}(\theta_\star^{(1)},X^{(t)}),\;\;\mbox{ and }\;\;\;\mathcal{G}_\tau^2 \eqdef \frac{1}{T}\sum_{t=\tau+1}^{T}\nabla \phi_t(\theta_\star^{(2)},X^{(t)}),\]
where $\nabla\phi_t(\theta,x)$ denotes the partial derivative of $u\mapsto \phi_t(u,x)$ at $\theta$. Also for $\tau\in\{1,\ldots,T-1\}$, and for $\theta\in\Theta$, we define,
\begin{multline*}
\mathcal{L}_1(\tau,\theta) \eqdef \frac{1}{T}\sum_{t=1}^{\tau}\left[\phi_t(\theta,X^{(t)}) - \phi_t(\theta_\star^{(1)},X^{(t)}) -\seq{\nabla \phi_t(\theta_\star^{(1)},X^{(t)}),\theta-\theta_\star^{(1)}}\right],\\
\mbox{ and }\;\;\;\mathcal{L}_2(\tau,\theta) \eqdef \frac{1}{T}\sum_{t=\tau+1}^{T}\left[\phi_t(\theta,X^{(t)}) - \phi_t(\theta_\star^{(2)},X^{(t)}) -\seq{\nabla \phi_t(\theta_\star^{(2)},X^{(t)}),\theta-\theta_\star^{(2)}}\right].
\end{multline*}

For $j=1,2$, define $\mathcal{A}_j\eqdef \left\{1\leq k\leq d:\theta_{\star k}^{(j)}\neq 0\right\}$, $s_j=|\mathcal{A}_j|$ , and 
\begin{equation}\label{eq:C:proof}
 \mathbb{C}_{j}\eqdef \left\{\theta\in\Theta:\; \sum_{k\in \mathcal{A}_j^c} |\theta_{k}^{(j)}|\leq 3\sum_{k\in\mathcal{A}_j}|\theta_{k}^{(j)}|\right\}.\end{equation}

The curvature of the function $\mathcal{L}_j(\tau,\cdot)$ is not always best described with the usual quadratic function $\theta\mapsto\|\theta-\theta_\star^{(j)}\|_2^2$. We will need a more flexible framework, in order  to handle $\mathcal{L}_j(\tau,\cdot)$ in the case of discrete Markov random fields. Let $\r:\;[0,\infty)\to [0,\infty)$ be continuous function such that $x\mapsto \r(x)/x$ is strictly increasing and $\lim_{x\downarrow 0} \r(x)/x=0$. We call $\r$ a rate function, and for $a>0$, we define $\Psi_\r(a)\eqdef \inf\{x>0:\; \r(x)/x \geq a\}$ ($\inf\emptyset=+\infty$). For $\tau\in\{1,\ldots,T-1\}$,  $\lambda>0$, a rate function $\r$, $c>0$, and for $j=1,2$ we work with the event 
\begin{multline*}
\mathcal{E}^j_\tau\left(\lambda,\r,c\right)\eqdef \left\{\|G^j_\tau\|_{\infty}\leq\frac{\lambda}{2},\;\;\;\inf_{\theta\neq \theta_{\star}^{(j)},\;\theta-\theta_\star^{(j)}\in\mathbb{C}_j}\frac{\mathcal{L}_j(\tau,\theta)}{\r\left(\|\theta-\theta_\star^{(j)}\|_2\right)}\geq \frac{\tau}{T},\right.\;\\
\;\;\left. \sup_{\theta\neq \theta_{\star}^{(j)},\;\theta-\theta_\star^{(j)}\in\mathbb{C}_j}\frac{\mathcal{L}_j(\tau,\theta)}{\|\theta-\theta_\star^{(j)}\|_2^2}\leq \frac{\tau}{T}\frac{c}{2}\right\}.\end{multline*}

Define 
\[\kappa_0^{(t)} \eqdef \left\{\begin{array}{ll}\PE\left[\phi_t(\theta_\star^{(2)},X^{(t)}) -\phi_t(\theta_\star^{(1)},X^{(t)})\right] & \mbox{ if } t\leq \tau_\star \\ \PE\left[\phi_t(\theta_\star^{(1)},X^{(t)}) -\phi_t(\theta_\star^{(2)},X^{(t)})\right] & \mbox{ if } t>\tau_\star\end{array}\right.,\]
and
\[U^{(t)} \eqdef \left\{\begin{array}{ll}\phi_t(\theta_\star^{(2)},X^{(t)}) -\phi_t(\theta_\star^{(1)},X^{(t)})-\kappa_0^{(t)} & \mbox{ if } t\leq \tau_\star \\ \phi_t(\theta_\star^{(1)},X^{(t)}) -\phi_t(\theta_\star^{(2)},X^{(t)})-\kappa_0^{(t)} & \mbox{ if } t>\tau_\star\end{array}\right..\]

We make the following assumption.
\begin{assumptionA}\label{A1}
There exist finite constants $\sigma_{0t}>0$ such that 
\[\PE\left(e^{xU^{(t)}}\right)\leq e^{x^2\sigma_{0t}^2\|\theta_\star^{(2)}-\theta_\star^{(1)}\|_2^2/2},\;\;\mbox{ for all }\;\;x>0.\]
Furthermore, there exist $B_0>0$, $\bar\sigma_0^2>0$, $\bar\kappa_0>0$ such that for all integer $k\geq B_0$,
\begin{equation}\label{eq1:A1}
\min\left(\frac{1}{k}\sum_{t=\tau_\star-k+1}^{\tau_\star}\kappa_0^{(t)},\;\frac{1}{k}\sum_{t=\tau_\star+1}^{\tau_\star+ k}\kappa_0^{(t)}\right)\geq \bar\kappa_0\|\theta_\star^{(2)}-\theta_\star^{(1)}\|_2^2,
\end{equation}
and
\begin{equation}\label{eq2:A1}
\max\left(\frac{1}{k}\sum_{t=\tau_\star-k+1}^{\tau_\star}\sigma_{0t}^2,\frac{1}{k}\sum_{t=\tau_\star+1}^{\tau_\star+ k}\sigma_{0t}^2\right)\leq \bar\sigma_0^2.
\end{equation}
\end{assumptionA}

\begin{theoremA}\label{thm:appendix}
Assume A\ref{A1}, and $\theta_\star^{(1)}\neq \theta_\star^{(2)}$. Suppose that $\hat\tau$ is defined over a search domain $\mathcal{T}\ni\tau_\star$, and with penalty $\lambda_{j,\tau}>0$ (for $j=1,2$). For $j=1,2$, take a rate function $\r_j$, constant $c_{j}>0$, and define $\mathcal{E} \eqdef \cap_{\tau\in\mathcal{T}} \mathcal{E}^1_\tau\left(\lambda_{1,\tau},\r_{1},c_{1}\right)\cap \mathcal{E}^2_\tau\left(\lambda_{2,\tau},\r_{2},c_{2}\right)$. Set
\begin{multline*}
\delta(\tau)\eqdef \Psi_{\r_1}\left(6\left(\frac{T}{\tau}\right)s_1^{1/2}\lambda_{1,\tau}\right)\left[2s_1^{1/2}T\lambda_{1,\tau} +\tau \Psi_{\r_1}\left(6\left(\frac{T}{\tau}\right)s_1^{1/2}\lambda_{1,\tau}\right)\right] \\
+\Psi_{\r_2}\left(6\left(\frac{T}{T-\tau}\right)s_2^{1/2}\lambda_{2,\tau}\right)\left[2s_2^{1/2}T\lambda_{2,\tau} +(T-\tau) \Psi_{\r_2}\left(6\left(\frac{T}{T-\tau}\right)s_2^{1/2}\lambda_{2,\tau}\right)\right],\end{multline*}
$\delta \eqdef \sup_{\tau\in\mathcal{T}}\delta(\tau)$, and $B \eqdef \max\left(B_0,\frac{4\delta}{\bar\kappa_0\|\theta_\star^{(2)}-\theta_\star^{(1)}\|_2^2}\right)$, with $B_0$ as in A\ref{A1}. Then 
\begin{equation}\PP\left(\left|\hat\tau-\tau_\star\right|>B\right)\leq 2\PP(\mathcal{E}^c)  + \frac{4\exp\left(-\frac{\bar\kappa_0^2 \delta}{2\bar\sigma_0^2}\right)}{1 -\exp\left(-\frac{\bar\kappa_0^2\|\theta_\star^{(2)} -\theta_\star^{(1)}\|_2^2}{8\bar\sigma_0^2}\right)}. \end{equation}
\end{theoremA}
\proof 
The starting point of the proof is the following variant of a result due to \cite{negh}.
\begin{Lemma}\label{techlemma 2:theta}
Fix $\tau\in\left\{1,2,\ldots,T-1\right\}$. On $\mathcal{E}^1_\tau\left(\lambda_{1,\tau},\r_{1},c_{1}\right)\cap \mathcal{E}^2_\tau\left(\lambda_{2,\tau},\r_{2},c_{2}\right)$, $\hat\theta_{j,\tau}-\theta_\star^{(j)}\in \mathbb{C}_j$, ($j=1,2$), where $\mathbb{C}_j$ is defined in (\ref{eq:C:proof}), and 
\begin{multline}
\label{result:lemma 2}
\|\hat{\mathbf{\theta}}_{1,\tau}-\mathbf{\theta}_\star^{(1)}\|_2\leq\Psi_{\r_1}\left(6\left(\frac{T}{\tau}\right)s_1^{1/2}\lambda_{1,\tau}\right),\\
\mbox{ and }\;\; \|\hat{\mathbf{\theta}}_{2,\tau}-\mathbf{\theta}_\star^{(2)}\|_2\leq\Psi_{\r_2}\left(6\left(\frac{T}{T-\tau}\right)s_2^{1/2}\lambda_{2,\tau}\right).
\end{multline}
\end{Lemma}
\proof 
We prove the first inequality. The second follows similarly. We set
\[\U(\theta) \eqdef \frac{1}{T}\sum_{t=1}^\tau \phi_t(\theta,X^{(t)}) + \lambda_{1,\tau}\|\theta\|_{1} -\left( \frac{1}{T}\sum_{t=1}^\tau \phi_t(\theta_\star^{(1)},X^{(t)}) + \lambda_{1,\tau}\|\theta_\star^{(1)}\|_{1}\right).\]
Since $\hat\theta_{1,\tau} =  \argmin_{\theta\in\Theta}\left[\frac{1}{T}\sum_{t=1}^\tau \phi_t(\theta,X^{(t)}) + \lambda_{1,\tau}\|\theta\|_{1}\right]$, and using the convexity of the functions $\phi_t$  we have
\[ 0\geq \U(\hat\theta_{1,\tau}) \geq \seq{G_\tau^1,\hat\theta_{1,\tau}-\theta_\star^{(1)}} +\lambda_{1,\tau}\left(\|\hat\theta_{1,\tau}\|_{1} - \|\theta_\star^{(1)}\|_{1}\right).\]
On $\mathcal{E}^1_\tau\left(\lambda_{1,\tau},\r_{1},c_{1}\right)$, $\|G_\tau^1\|_\infty\leq \lambda_{1,\tau}/2$. Using this and some easy algebra as in \cite{negh}, shows that $\hat\theta_{1,\tau}-\theta_\star^{(1)}\in \mathbb{C}_1$.
Set $b=\Psi_{\r_1}\left(6\left(\frac{T}{\tau}\right)s_1^{1/2}\lambda_{1,\tau}\right)$. We will show that for all $\theta\in\rset^d$ such that $\theta-\theta_\star^{(1)}\in\mathbb{C}_1$, and $\|\theta-\theta_\star^{(1)}\|_2>b$, we have $\U(\theta)>0$. Since $\U(\hat\theta_{1,\tau})\leq 0$, and $\hat \theta_{1,\tau}-\theta_\star^{(1)}\in\mathbb{C}_1$, the claim that $\|\theta-\theta_\star^{(1)}\|_2\leq b$ follows. On the event $\mathcal{E}^1_\tau\left(\lambda_{1,\tau},\r_{1},c_{1}\right)$, and for $\theta-\theta_{\star}^{(1)}\in\mathbb{C}_1$, we have
\begin{eqnarray*}
\U(\theta) &=& \seq{G_\tau^1,\theta-\theta_\star^{(1)}} +\mathcal{L}_1(\tau,\theta)  + \lambda_{1,\tau}\left(\|\theta\|_1 - \|\theta_\star^{(1)}\|_{1}\right)\\
& \geq & \frac{\tau}{T}\r_1(\|\theta-\theta_\star^{(1)}\|_2) -\frac{3\lambda_{1,\tau}}{2}\|\theta-\theta_\star^{(1)}\|_1\\
&\geq & \frac{\tau}{T} \left[\r_1(\|\theta-\theta_\star^{(1)}\|_2) - 6\left(\frac{T}{\tau}\right)s_1^{1/2}\lambda_{1,\tau}\|\theta-\theta_\star^{(1)}\|_2\right].\end{eqnarray*}
Using the definition of $\Psi_{\r_1}$, we then see that $\U(\theta)>0$ for $\|\theta-\theta_\star^{(1)}\|_2> b$. This ends the proof.
\qed

\medskip
The next result follows easily.
\begin{Lemma}\label{bound:remainder term}
Fix $\tau\in\left\{1,2,\ldots,T-1\right\}$. On $\mathcal{E}^1_\tau\left(\lambda_{1,\tau},\r_{1},c_{1}\right)\cap \mathcal{E}^2_\tau\left(\lambda_{2,\tau},\r_{2},c_{2}\right)$,
\[\left|\ell_T(\tau,\hat\theta_{1,\tau},\hat\theta_{2,\tau}) - \ell_T(\tau,\theta_\star^{(1)},\theta_\star^{(2)})\right| \leq  \frac{\delta(\tau)}{T},\]
where 
\begin{multline*}
\delta(\tau)\eqdef \Psi_{\r_1}\left(6\left(\frac{T}{\tau}\right)s_1^{1/2}\lambda_{1,\tau}\right)\left[2s_1^{1/2}T\lambda_{1,\tau} +\frac{\tau c_1}{2} \Psi_{\r_1}\left(6\left(\frac{T}{\tau}\right)s_1^{1/2}\lambda_{1,\tau}\right)\right] \\
+\Psi_{\r_2}\left(6\left(\frac{T}{T-\tau}\right)s_2^{1/2}\lambda_{2,\tau}\right)\left[2s_2^{1/2}T\lambda_{2,\tau} +\frac{(T-\tau)c_2}{2} \Psi_{\r_2}\left(6\left(\frac{T}{T-\tau}\right)s_2^{1/2}\lambda_{2,\tau}\right)\right].\end{multline*}
\end{Lemma}
\begin{proof} 
\begin{multline*}
\ell_T(\tau,\hat\theta_{1,\tau},\hat\theta_{2,\tau}) - \ell_T(\tau,\theta_\star^{(1)},\theta_\star^{(2)}) = \frac{1}{T}\sum_{t=1}^\tau \left[\phi_t(\hat\theta_{1,\tau},X^{(t)}) - \phi_t(\theta_\star^{(1)},X^{(t)})\right]\\
+ \frac{1}{T}\sum_{t=\tau+1}^T \left[\phi_t(\hat\theta_{2,\tau},X^{(t)}) - \phi_t(\theta_\star^{(2)},X^{(t)})\right].\end{multline*}
From the definition
\[\frac{1}{T}\sum_{t=1}^\tau \left[\phi_t(\hat\theta_{1,\tau},X^{(t)})-\phi_t(\theta_\star^{(1)},X^{(t)})\right] = \seq{G_\tau^1,\hat\theta_{1,\tau}-\theta_\star^{(1)}} +\mathcal{L}_1(\tau,\hat\theta_{1,\tau}).\]
On $\mathcal{E}^1_\tau\left(\lambda_{1,\tau},\r_{1},c_{1}\right)$, and using Lemma \ref{techlemma 2:theta}, we have
\[\left|\seq{G_\tau^1,\hat\theta_{1,\tau}-\theta_\star^{(1)}}\right|\leq \frac{\lambda_{1,\tau}}{2}\|\hat\theta_{1,\tau}-\theta_\star^{(1)}\|_1 \leq 2s_1^{1/2}\lambda_{1,\tau}\Psi_{\r_1}\left(6\left(\frac{T}{\tau}\right)s_1^{1/2}\lambda_{1,\tau}\right),\]
and
\[\mathcal{L}_1(\tau,\hat\theta_{1,\tau}) \leq \frac{\tau}{T}\frac{c_{1}}{2}\|\hat\theta_{1,\tau}-\theta_\star^{(1)}\|_2^2\leq \frac{\tau c_1}{2T}\Psi_{\r_1}\left(6\left(\frac{T}{\tau}\right)s_1^{1/2}\lambda_{1,\tau}\right)^2.\]
Hence
\begin{multline*}
\left|\frac{1}{T}\sum_{t=1}^\tau \left[\phi_t(\hat\theta_{1,\tau},X^{(t)})-\phi_t(\theta_\star^{(1)},X^{(t)})\right]\right| \\
\leq \frac{1}{T}\Psi_{\r_1}\left(6\left(\frac{T}{\tau}\right)s_1^{1/2}\lambda_{1,\tau}\right)\left[2s_1^{1/2}T\lambda_{1,\tau} +\frac{\tau c_1}{2} \Psi_{\r_1}\left(6\left(\frac{T}{\tau}\right)s_1^{1/2}\lambda_{1,\tau}\right)\right].
\end{multline*}
A similar bound holds for the second term, and the lemma follows easily.
\qed

\medskip
\end{proof}
We are now in position to prove Theorem \ref{thm:appendix}. We have
\[\PP\left(\left|\hat\tau-\tau_\star\right|>B\right) = \PP\left(\hat\tau>\tau_\star +B\right) + \PP\left(\hat\tau<\tau_\star-B\right).\]
We bound the first term $\PP\left(\hat\tau>\tau_\star +B\right)$. The second term follows similarly by working with the reversed sequence $X^{(T)},\ldots,X^{(1)}$.

For $\tau>\tau_\star$, we shall use $\ell_T\left(\tau\right)$ instead of $\ell_T\left(\tau;\hat{\mathbf{\theta}}_{1,\tau},\hat{\mathbf{\theta}}_{2,\tau}\right)$ for notational convenience, and we define $r_T\left(\tau\right)\eqdef\ell_T(\tau)-\ell_T\left(\tau,\theta_\star^{(1)},\theta_\star^{(2)}\right)$.  We have 
\begin{eqnarray*}\ell_T\left(\tau\right) &=&\ell_T\left(\tau,\theta_\star^{(1)},\theta_\star^{(2)}\right) + r_T(\tau),\\
&=&\left[\ell_T\left(\tau,\theta_\star^{(1)},\theta_\star^{(2)}\right) -\ell_T\left(\tau_\star,\theta_\star^{(1)},\theta_\star^{(2)}\right)\right] +\ell_T\left(\tau_\star,\theta_\star^{(1)},\theta_\star^{(2)}\right) + r_T(\tau).
\end{eqnarray*}
Hence
\begin{equation}\label{proof:thm1:eq1}
\ell_T(\tau)- \ell_T(\tau_\star)  = \left[\ell_T\left(\tau,\theta_\star^{(1)},\theta_\star^{(2)}\right) -\ell_T\left(\tau_\star,\theta_\star^{(1)},\theta_\star^{(2)}\right)\right] + r_T(\tau) -r_T(\tau_\star).\end{equation}
It is straightforward to check that for $\tau>\tau_\star$,
\[\ell_T\left(\tau,\theta_\star^{(1)},\theta_\star^{(2)}\right) -\ell_T\left(\tau_\star,\theta_\star^{(1)},\theta_\star^{(2)}\right)=\frac{1}{T}\sum_{t=\tau_\star+1}^\tau\left(\phi_t(\theta_\star^{(1)},X^{(t)})-\phi_t(\theta_\star^{(2)},X^{(t)})\right).\]
Therefore, and using the definition of $U^{(t)}$ and $\kappa_0^{(t)}$, (\ref{proof:thm1:eq1}) becomes 
\begin{equation}
\ell_T(\tau)- \ell_T(\tau_\star) = \frac{1}{T}\sum_{t=\tau_\star+1}^\tau \kappa_0^{(t)} + \frac{1}{T}\sum_{t=\tau_\star+1}^\tau U^{(t)} + r_T(\tau)-r_T(\tau_\star).
\end{equation}
We conclude from Lemma \ref{bound:remainder term} that on the event $\mathcal{E}$,
\begin{multline}\label{proof:thm1eq2}
\ell_T(\tau)- \ell_T(\tau_\star)=\frac{1}{T}\sum_{t=\tau_\star+1}^\tau \kappa_0^{(t)} + \frac{1}{T}\sum_{t=\tau_\star+1}^\tau U^{(t)} + \epsilon_T(\tau),\;\;\\
\mbox{ where }\; |\epsilon_T(\tau)|\leq \frac{2\sup_{\tau\mathcal{T}}|\delta(\tau)|}{T}=\frac{2\delta}{T}.\end{multline}
Therefore,
\begin{equation*}
\PP\left(\hat\tau>\tau+B\right) \leq \PP(\mathcal{E}^c) + \sum_{j\geq 0,\;\tau_\star+ \lceil B\rceil+j\in\mathcal{T} }\PP\left(\mathcal{E},\; \hat\tau =\tau_\star+ \lceil B\rceil+j\right).
\end{equation*}
Using (\ref{proof:thm1eq2}), we have
\begin{eqnarray*}
\PP\left(\mathcal{E},\; \hat\tau =\tau_\star+ \lceil B\rceil+j\right)& \leq & \PP\left(\mathcal{E},\; \ell_T(\tau_\star+ \lceil B\rceil+j)\leq \ell_T(\tau_\star)\right)\\
& \leq & \PP\left(\left|\sum_{t=\tau_\star+1}^{\tau_\star+\lceil B\rceil+j} U^{(t)}\right|> \sum_{t=\tau_\star+1}^{\tau_\star+\lceil B\rceil+j}\kappa_0^{(t)} -2\delta\right).
\end{eqnarray*}
However, since $B>B_0$, by Assumption A\ref{A1}, 
\[\sum_{t=\tau_\star+1}^{\tau_\star+\lceil B\rceil+j}\kappa_0^{(t)} -2\delta \geq \left(\lceil B\rceil+j\right)\bar\kappa_0 \|\theta_\star^{(2)} -\theta_\star^{(1)}\|_2^2 -2\delta \geq \frac{1}{2}\left(\lceil B\rceil+j\right)\bar\kappa_0 \|\theta_\star^{(2)} -\theta_\star^{(1)}\|_2^2.\]
The first part of A\ref{A1} implies that the random variables $Z^{(t)}$ are sub-Gaussian, and by standard exponential bounds for sub-Gaussian random variables, we then have
\begin{eqnarray*}
\PP\left[\mathcal{E},\; \ell_T(\tau_\star+ \lceil B\rceil+j)\leq \ell_T(\tau_\star)\right] &\leq & 2\exp\left(-\frac{\left(\lceil B\rceil+j\right)^2\bar\kappa_0^2 \|\theta_\star^{(2)} -\theta_\star^{(1)}\|_2^4}{8\|\theta_\star^{(2)}-\theta_\star^{(1)}\|_2^2\sum_{t=\tau_\star+1}^{\tau_\star+\lceil B\rceil+j}\sigma_{0t}^2}\right),\\
&\leq & 2\exp\left(-\frac{\left(\lceil B\rceil+j\right)\bar\kappa_0^2 \|\theta_\star^{(2)} -\theta_\star^{(1)}\|_2^2}{8\bar \sigma_0^2}\right),
\end{eqnarray*}
where the last inequality uses (\ref{eq2:A1}). We can conclude that 
\begin{eqnarray}
\label{main prob:thm 1}
\mathbb{P}\left[\hat{\tau} > \tau_\star+B\right]&\leq & \PP(\mathcal{E}^c)  +2\sum_{j\geq 0} \exp\left(-\frac{\left(\lceil B\rceil+j\right)\bar\kappa_0^2 \|\theta_\star^{(2)} -\theta_\star^{(1)}\|_2^2}{8\bar \sigma_0^2}\right)\nonumber\\
&\leq & \PP(\mathcal{E}^c)  + 2\frac{\exp\left(-\frac{B\bar\kappa_0^2 \|\theta_\star^{(2)} -\theta_\star^{(1)}\|_2^2}{8\bar\sigma_0^2}\right)}{1 -\exp\left(-\frac{\bar\kappa_0^2\|\theta_\star^{(2)} -\theta_\star^{(1)}\|_2^2}{8\bar\sigma_0^2}\right)},
\end{eqnarray}
as claimed.
\qed

\section{Proof of Theorem \ref{thm1}}
We will deduce Theorem \ref{thm1} from Theorem \ref{thm:appendix}. We take $\Theta$ as $\M_p$, the set of all $p\times p$ real symmetric matrices, equipped with the (modified) Frobenius inner product $\seqF{\theta,\vartheta}\eqdef\sum_{k\leq j}\theta_{jk}\vartheta_{jk}$, and the associated norm $\normF{\theta}\eqdef \sqrt{\seq{\theta,\theta}}$. With this inner product, we identify $\M_p$ with the Euclidean space $\rset^{d}$, with $d=p(p+1)/2$. This puts us in the setting of Theorem \ref{thm:appendix}. 

We will use the following notation. If $u\in\rset^q$, for some integer $q\geq 1$, and $A$ is an ordered subset of $\{1,\ldots,q\}$, we define $u_A\eqdef (u_j,\,j\in A)$, and $u_{-j}$ is a shortcut for $u_{\{1,\ldots,q\}\setminus\{j\}}$.  We define the function $B_{jk}(x,y)=B_0(x)$ if $j=k$, and $B_{jk}(x,y)=B(x,y)$ if $j\neq k$.

In the present case, the function $\phi_t$ is $\phi$ as given in (\ref{def:phi}), and does not depend on $t$. The following properties of the conditional distribution (\ref{full:cond}) will be used below. It is well known (and easy to prove using Fisher's identity)  that the function $\theta\mapsto \phi(\theta,x)$ is Lispchitz and 
\begin{equation}\label{lip:phi}
\left|\phi(\theta,x)-\phi(\vartheta,x)\right|\leq 2c_0\|\theta-\vartheta\|_1,\;\;\theta,\vartheta\in\M_p,\;x\in\Xset^p,\end{equation}
where $c_0$ is as in (\ref{def:c0}). From the expression (\ref{full:cond}) of the conditional densities, using straightforward algebra, it is easy to show that the negative log-pseudo-likelihood function $\phi(\theta,x)$ satisfies the following. For all $\theta,\Delta\in\M_p$, and $x\in\Xset^p$,
\begin{multline}\label{exp:phi}
\phi(\theta+\Delta,x)-\phi(\theta,x)-\seqF{\nabla_\theta\phi(\theta,x),\Delta} \\
=\sum_{j=1}^p\left[\log Z^{(j)}_{\theta+\Delta}(x)-\log Z^{(j)}_{\theta}(x)-\sum_{k=1}^p \Delta_{jk}\frac{\partial}{\partial \theta_{jk}}\log Z^{(j)}_{\theta}(x)\right].\end{multline}
Furthermore by   Taylor expansion, we have
\begin{multline}\label{boundZ}
\log Z^{(j)}_{\theta+\Delta}(x)-\log Z^{(j)}_{\theta}(x)-\sum_{k=1}^p \Delta_{jk}\frac{\partial}{\partial \theta_{jk}}\log Z^{(j)}_{\theta}(x)\\
=\int_0^1(1-t)\textsf{Var}_{\theta+t\Delta}\left(\sum_{k=1}^p\Delta_{jk}B_{jk}(X_j,X_k)\vert X_{-j}\right) \rmd t\leq \frac{c_0^2}{2}\left(\sum_{k=1}^p|\Delta_{jk}|\right)^2.\end{multline}
By the self-concordant bound derived in \cite{atc}~Lemma A2, we have
\begin{multline}\label{lower:boundZ}
\log Z^{(j)}_{\theta+\Delta}(x)-\log Z^{(j)}_{\theta}(x)-\sum_{k=1}^p \Delta_{jk}\frac{\partial}{\partial \theta_{jk}}\log Z^{(j)}_{\theta}(x)\\
\geq \frac{1}{2+c_0\sum_{k=1}^p|\Delta_{jk}|}\textsf{Var}_{\theta}\left(\sum_{k=1}^p\Delta_{jk}B_{jk}(X_j,X_k)\vert X_{-j}\right).\end{multline}

\medskip

\proof[Proof of Theorem \ref{thm1}]
Let us first show that under assumption H\ref{H3} of Theorem \ref{thm1}, A\ref{A1} holds. Since in this case $\phi_t$ does not actually depend on $t$, we can take $B_0=1$ in A\ref{A1}, and (\ref{eq1:A1}) follows automatically from H\ref{H3} with $\bar\kappa_0=\kappa/\|\theta_\star^{(2)}-\theta_\star^{(1)}\|_2^2$. Also, (\ref{lip:phi}) implies that $|U^{(t)}|\leq 4 c_0\|\theta_\star^{(2)}-\theta_\star^{(1)}\|_1\leq 4c_0 s^{1/2} \|\theta_\star^{(2)}-\theta_\star^{(1)}\|_2$, where $s$ denotes the number of non-zero entries of $\theta^{(2)}-\theta_\star^{(1)}$. Hence for all $x>0$,
\[\PE\left(e^{x U^{(t)}}\right)\leq \exp\left(8x^2c_0^2s\|\theta_\star^{(2)}-\theta_\star^{(1)}\|_2^2\right).\]
This establishes the sub-Gaussian condition of A\ref{A1}, and (\ref{eq2:A1}) holds with $\bar\sigma_0^2 = 16c_0^2s$.

\medskip

For $j=1,2$, let $\lambda_{1,\tau}$, $\lambda_{2,\tau}$ as  in (\ref{lambda1:lambda2}). We will apply Theorem \ref{thm:appendix} with  $c_j = 64c_0s_j$, the rate function $\r_j(x) = \frac{\rho_j x^2}{2+4c_0s_j^{1/2}x}$, $x>0$, and with the event $\mathcal{E} = \bigcap_{\tau\in\Tau}\left[\mathcal{E}^1_\tau\left(\lambda_{1,\tau},\r_1,c_1\right)\cap\mathcal{E}^2_\tau\left(\lambda_{2,\tau},\r_2,c_2\right)\right]$, where the search domain $\mathcal{T}$ satisfies (\ref{cond H2:Tau_p}), (\ref{cond H2:Tau_m}), and (\ref{cond H2:Tau_tau}).  Notice that if $\r(x) = \rho x^2/(2 + bx)$, $\rho,b>0$, is a rate function, then for $a>0$, $\Psi_\r(a)\eqdef \inf\{x>0:\;r(x)\geq ax\} \leq 4a/\rho$, provided that $2ba\leq \rho$. Hence
\[\Psi_{\r_1}\left(6\left(\frac{T}{\tau}\right)s_1^{1/2}\lambda_{1,\tau}\right) \leq \frac{4}{\rho_1}6\left(\frac{T}{\tau}\right)s_1^{1/2}\lambda_{1,\tau}=24\times 32 c_2\frac{s_1^{1/2}}{\rho_1}\sqrt{\frac{\log(dT)}{\tau}},\]
provided that $\tau\geq (48\times 32)^2c_0^2 \left(\frac{s_1}{\rho_1}\right)^2\log(dT)$. Therefore, given that all $\tau\in\Tau$ satisfies (\ref{cond H2:Tau_tau}), with some simple algebra we see that there exists a universal constant $a$ that we can take as $a=(24\times 32\times 64)^2$, such that for all $\tau\in\Tau$,
\[\delta(\tau)\leq \delta = ac_0^2M\log(dT),\]
where 
\[M = \left[\frac{s_1}{\rho_1}\left(1+c_0\frac{s_1}{\rho_1}\right) + \frac{s_2}{\rho_2}\left(1+ c_0\frac{s_2}{\rho_2}\right)\right].\]
Therefore in Theorem \ref{thm:appendix}, we can take $B = \frac{4ac_0^2M\log(dT)}{\kappa}$, and by the conclusion of Theorem \ref{thm:appendix},
\[\PP\left[\left|\hat\tau-\tau_\star\right|>B\right]\leq 2\PP(\mathcal{E}^c) + \frac{4\exp\left(-\frac{\delta}{32c_0^2s}\left(\frac{\kappa}{\|\theta_\star^{(2)} - \theta_\star^{(1)}\|^2_2}\right)^2\right)}{1-\exp\left(-\frac{\kappa^2}{2^7c_0^2s\|\theta_\star^{(2)} - \theta_\star^{(1)}\|^2_2}\right)}.\]
We show in Lemma \ref{bound:lambda} and Lemma \ref{bound:E} below that $\PP(\mathcal{E}^c)\leq 8/d$, and this ends the proof.

\qed

\begin{Lemma} \label{bound:lambda}
Let $\lambda_{1,\tau},\lambda_{2,\tau}$ be as in equation ~\eqref{lambda1:lambda2}. Suppose that the search domain $\mathcal{T}$ is such that (\ref{cond H2:Tau_p})-(\ref{cond H2:Tau_m}) hold. Then 
\[ 
\mathbb{P}\left[\max_{\tau\in\Tau} \lambda^{-1}_{1,\tau}\big\|G^1_\tau\big\|_{\infty}> \frac{1}{2}\right]\leq\frac{2}{d},\;\;\mbox{ and }\;\; \mathbb{P}\left[\max_{\tau\in\Tau} \lambda^{-1}_{2,\tau}\big\|G^2_\tau\big\|_{\infty}> \frac{1}{2}\right]\leq\frac{2}{d},\]
 where $d=p(p+1)/2$.
\end{Lemma}
\begin{proof}
We carry the details for the first bound. The second is done similarly by working with the reversed sequence $X^{(T)},\ldots,X^{(1)}$. Fix $1\leq j\leq i\leq p$,  $t\in\Tau$, and define $V_{ij}^{(t)}\eqdef\frac{\partial}{\partial \theta_{ij}}\phi(\theta_\star^{(1)},X^{(t)})$. We calculate that 
\[
V_{ij}^{(t)}=\left\{\begin{array}{lc}
-B_0(X_i^{(t)}) +\PE_{\theta_\star^{(1)}}(B_0(X_i\vert X_{-i}^{(t)}) & \mbox{ if } i=j\\
-2B(X^{(t)}_i,X^{(t)}_j)+\PE_{\theta_\star^{(1)}}\left( B(X_i,X_j^{(t)})\vert X_{-i}^{(t)}\right)+\PE_{\theta_\star^{(1)}}\left( B(X_i,X_j^{(t)})\vert X_{-j}^{(t)}\right) & \mbox{ if }j<i.\end{array}\right.
\]
In the above display the notation $\PE_{\theta_\star^{(1)}}\left( B(X_i,X_j^{(t)})\vert X_{-i}^{(t)}\right)$ is defined as the function $z\mapsto\PE_{\theta_\star^{(1)}}\left( B(X_i,z_j)\vert X_{-i}=z_{-i}\right)$  evaluated on $X^{(t)}$.  Since $X^{(1:\tau_\star)}\stackrel{i.i.d}{\sim} g_{\theta_\star^{(1)}}$, it follows that $\PE(V_{ij}^{(t)})=0$ for $t=1,\ldots,\tau_\star$. We set $\mu_{ij} \eqdef \PE(V_{ij}^{(\tau_\star+1)})=\PE(V_{ij}^{(t)})$ for $t=\tau_\star+1,\ldots,T$. We also set $\bar V_{ij}^{(t)} \eqdef V_{ij}^{(t)}-\PE\left(V_{ij}^{(t)}\right)$. It is easy to see that $|\bar V_{ij}^{(t)}|\leq 4c_0$, where $c_0$ is defined in (\ref{def:c0}) . With these notations, for $\tau\in\Tau$, we can write 
\[(G_\tau^1)_{ij} = \frac{1}{T}\sum_{t=1}^\tau \bar V_{ij}^{(t)} + \frac{(\tau-\tau_\star)_+\mu_{ij}}{T},\]
where $a_+\eqdef \max(a,0)$. For $t>\tau_\star$, Lemma \ref{lemlip} can be used to write
 \begin{multline*}
 \left|\PE\left[B(X^{(t)}_i,X^{(t)}_j)-\PE_{\theta_\star^{(1)}}\left( B(X_i,X_j^{(t)})\vert X_{-i}^{(t)}\right)\right]\right|\\
 =\left|\PE\left[\int_\Xset B(u,X_j^{(t)})f_{\theta_\star^{(2)}}(u\vert X_{-i}^{(t)})\rmd u - \int_\Xset B(u,X_j^{(t)})f_{\theta_\star^{(1)}}(u\vert X_{-i}^{(t)})\rmd u\right]\right|\\
 \leq c_0^2\sum_{j=1}^p|\theta_{\star,ij}^{(2)}- \theta_{\star,ij}^{(1)}|\leq b c_0^2,\end{multline*}
where $b$ is as in (\ref{def:b}). 
Hence
 \begin{equation*}
\left|\mu_{ij}\right|\leq 2\max_{j\leq i}\left |\PE_{\theta_\star^{(2)}}\left[ B(X_i^{(t)},X^{(t)}_j)-\PE_{\theta_\star^{(1)}}\left( B(X_i^{(t)},X^{(t)}_j)\vert X_{-j}^{(t)}\right)\right]\right|\leq 2 b c_0^2.\]
Set $\lambda_\tau \eqdef (A\sqrt{\tau}/T)$, where 
\[A \eqdef 32c_0\sqrt{\log(dT)}.\]

By a union-bound argument,
\begin{multline}\label{eq:union_bound}
\PP\left[\max_{\tau\in\Tau}2\lambda_{\tau}^{-1}\|G_{\tau}^1\|_{\infty} > 1\right] \\
\leq \sum_{\tau\in\mathcal{T}}\sum_{i,j}\PP\left[\frac{1}{A\sqrt{\tau}}\left|\sum_{t=1}^\tau \bar V_{ij}^{(t)}\right| + \frac{2bc_0^2(\tau-\tau_\star)_+}{A\sqrt{\tau}} > \frac{1}{2}\right].\end{multline}
Since $A = 32c_0\sqrt{\log(dT)}$, for $\tau\in \mathcal{T}$, and using (\ref{cond H2:Tau_p}) we see that $\max_{\tau\in\Tau} \frac{2bc_0^2(\tau-\tau_\star)_+}{A\sqrt{\tau}}\leq 1/4$. Hence
\begin{eqnarray}\label{eq:expo_bound_thm1}
\PP\left[\max_{\tau\in\Tau}2\lambda_{\tau}^{-1}\|G_{\tau}^1\|_{\infty} > 1\right] &\leq & \sum_{\tau\in\mathcal{T}}\sum_{i,j}\PP\left[\left|\sum_{t=1}^\tau \bar V_{ij}^{(t)}\right| > \frac{A\sqrt{\tau}}{4}\right],\\
& \leq & 2\sum_{\tau\in\mathcal{T}}\sum_{i,j}\exp\left(-\frac{A^2}{8^3c_0^2}\right)\leq \frac{2}{d}.\nonumber
\end{eqnarray}
where the second inequality uses Hoeffding's inequality.
\qed
\end{proof}

\begin{remark}
The $\log(dT)$ term that appears in the convergence rate of Theorem \ref{thm1} follows from the union bound and the exponential bound used in (\ref{eq:union_bound}), and (\ref{eq:expo_bound_thm1}) respectively. Alternatively, it is easy to see that one could also write
\[\PP\left[\max_{\tau\in\Tau}2\lambda_{\tau}^{-1}\|G_{\tau}^1\|_{\infty} > 1\right] \leq \sum_{i,j}\PP\left[\max_{\tau\in\Tau} \left|\frac{1}{\sqrt{\tau}}\sum_{t=1}^\tau \bar V_{ij}^{(t)}\right| > \frac{A}{4}\right].\]
Hence whether one can remote the $\log(T)$ term hinges on the existence of an exponential bound for the term $\max_{\tau\in\Tau} \left|\tau^{-1/2}\sum_{t=1}^\tau \bar V_{ij}^{(t)}\right|$.  Unfortunately we are not aware of any such result in the literature. The closest results available deal with the unweighted sums: $\max_{\tau\in\Tau} \left|\sum_{t=1}^\tau \bar V_{ij}^{(t)}\right|$ (see for instance \cite{pinelis} for some of the best bounds available).
\end{remark}

\begin{Lemma}\label{bound:E}
Assume H\ref{H1} and H\ref{H2}. Let $\lambda_{1,\tau}$ and $\lambda_{2,\tau}$ as in Equation~\eqref{lambda1:lambda2}, and let the search domain $\Tau$ be such that Equations~\eqref{cond H2:Tau_p}-\eqref{cond H2:Tau_m} hold. Take $c_1 = 64c_0s_1$, $c_2=64c_0s_2$ and
\[\r_1(x)=\frac{\rho_1 x^2}{ 2+4c_0s_1^{1/2}x}, \;\; \mbox{ and }\;\; \r_2(x)=\frac{\rho_2 x^2}{ 2+4c_0s_2^{1/2}x},\;\;x\geq 0.\]
Then the event $\bigcap_{\tau\in\Tau}\left[\mathcal{E}^1_\tau\left(\lambda_{1,\tau},\r_1,c_1\right)\cap\mathcal{E}^2_\tau\left(\lambda_{2,\tau},\r_2,c_2\right)\right]$ holds with probability at least $1-\frac{8}{d}$.
\end{Lemma}
\begin{proof}
We have seen in Lemma \ref{bound:lambda} that with $\lambda_{1,\tau}$ and $\lambda_{2,\tau}$ as in equation ~\eqref{lambda1:lambda2}, the event $\cap_{\tau\in\Tau}\left[\{\|G_\tau^1\|_\infty\leq \lambda_{1,\tau}/2\}\cap \{\|G_\tau^1\|_\infty\leq \lambda_{2,\tau}/2\}\right]$ holds with probability at least $1-2/d$. We have
\[\mathcal{L}_1(\tau,\theta) \eqdef \frac{1}{T}\sum_{t=1}^{\tau}\left[\phi(\theta,X^{(t)}) - \phi(\theta_\star^{(1)},X^{(t)}) -\seq{\nabla \phi(\theta_\star^{(1)},X^{(t)}),\theta-\theta_\star^{(1)}}\right].\]
(\ref{boundZ}) then implies that for all $\tau\in\mathcal{T}$, and $\theta-\theta_\star^{(1)}\in\mathbb{C}_1$,
\[
\mathcal{L}_1(\tau,\theta) \leq \frac{\tau}{T}\frac{4c_0^2}{2}\|\theta-\theta_\star^{(1)}\|_1^2\leq \frac{\tau}{T}\frac{64c_0^2s_1}{2}\|\theta-\theta_\star^{(1)}\|_2^2.\]
A similar bound holds for $j=2$. Hence $\cap_{\tau\in\Tau}\cap_{j=1}^2\left\{\sup_{\theta\neq \theta_{\star}^{(j)},\;\theta-\theta_\star^{(j)}\in\mathbb{C}_j}\frac{\mathcal{L}_j(\tau,\theta)}{\|\theta-\theta_\star^{(j)}\|_2^2}\leq \frac{\tau}{T}\frac{c_j}{2}\right\}$ holds with probability one.

Using (\ref{lower:boundZ}), we have
\begin{multline}\label{eq1:proof:boundE}
\mathcal{L}_1(\tau,\theta) \geq \frac{\tau}{T}\frac{1}{2+4c_0s_1^{1/2}\|\theta-\theta_\star^{(1)}\|_2}\\
\times \frac{1}{\tau}\sum_{t=1}^\tau \sum_{j=1}^p \textsf{Var}_{\theta_\star^{(1)}}\left(\sum_{k=1}^pB_{kj}(X_j^{(t)},X_k^{(t)})\left(\theta_{kj}-\theta_{\star,kj}^{(1)}\right)\vert X^{(t)}_{-j}\right).\end{multline}
We will now show that for all $\tau\in\mathcal{T}$, and all $\theta-\theta_\star^{(1)}\in\mathbb{C}_1$, with probability at least $1-2/d$, we have
\[\frac{1}{\tau}\sum_{t=1}^\tau \sum_{j=1}^p \textsf{Var}_{\theta_\star^{(1)}}\left(\sum_{k=1}^pB_{kj}(X_j^{(t)},X_k^{(t)})\left(\theta_{kj}-\theta_{\star,kj}^{(1)}\right)\vert X^{(t)}_{-j}\right)\geq \rho_1\|\theta-\theta_\star^{(1)}\|_2^2.\]
Given (\ref{eq1:proof:boundE}), this assertion will implies that $\mathcal{L}_1(\tau,\theta)\geq \frac{\tau}{T}\r_1(\|\theta-\theta_\star^{(1)}\|_2)$ for all $\theta-\theta_\star^{(1)}\in\mathbb{C}_1$ with probability at least $1-2/d$, where $\r_1(x) = \rho_1 x^2/(2+4c_0s_1^{1/2} x)$. The lemma will then follow easily.

For $\Delta\in\M_p$, we define
\[\mathcal{V}^1\left(\tau,\Delta\right) \eqdef \frac{1}{\tau}\sum_{t=1}^\tau \sum_{j=1}^p \textsf{Var}_{\theta_\star^{(1)}}\left(\sum_{k=1}^pB_{kj}(X_j^{(t)},X_k^{(t)})\Delta_{kj}\vert X^{(t)}_{-j}\right),\]
and
\begin{multline*}
 W^{(t)}_{jkk^{\prime}}\stackrel{\text{def}}{=}\textsf{Cov}_{\theta_\star^{(1)}}\left(B(X_j^{(t)},X^{(t)}_k),B(X^{(t)}_j,X^{(t)}_{k'})\vert X^{(t)}_{-j}\right)\\
  - \PE\left[\textsf{Cov}_{\theta_\star^{(1)}}\left(B(X^{(t)}_j,X^{(t)}_k),B(X^{(t)}_j,X^{(t)}_{k'})\vert X^{(t)}_{-j}\right)\right].\end{multline*}
Then for  $\Delta\in\mathbb{C}_1\setminus\{0\}$, 
\begin{eqnarray}
\mathcal{V}^1\left(\tau,\Delta\right) &=& 
\frac{1}{\tau}\displaystyle\sum_{t=1}^\tau\sum_{j=1}^p\sum_{k,k^{\prime}=1}^p\Delta_{jk}\Delta_{jk^{\prime}}
\PE\left[\textsf{Cov}_{\theta_\star^{(1)}}\left(B(X^{(t)}_j,X^{(t)}_k),B(X^{(t)}_j,X^{(t)}_{k'})\vert X^{(t)}_{-j}\right)\right]. \nonumber\\
&&\quad  +\frac{1}{\tau}\displaystyle\sum_{t=1}^\tau\sum_{j=1}^p\sum_{k,k^{\prime}=1}^p\Delta_{jk}\Delta_{jk^{\prime}}W^{(t)}_{jkk^{\prime}}
\end{eqnarray}
Using H\ref{H1}, we deduce that 
\begin{multline}
\label{eqn 1:lemma 5}
\mathcal{V}^1\left(\tau,\Delta\right)\geq 2\rho_1\|\Delta\|^2_2 + \frac{1}{\tau}\displaystyle\sum_{t=1}^\tau\sum_{j=1}^p\sum_{k,k^{\prime}=1}^p\Delta_{jk}\Delta_{jk^{\prime}}W^{(t)}_{jkk^{\prime}} \\
+\frac{(\tau-\tau_\star)_+}{\tau}\displaystyle\sum_{j=1}^p\mathbb{E}_{\mathbf{\theta}_\star^{(2)}}\left[\textsf{Var}_{\theta_\star^{(1)}}\left(\sum_{k=1}^p\Delta_{jk}B_{ik}(X_j,X_k)\vert X_{-j}\right)\right]\\
-\frac{(\tau-\tau_\star)_+}{\tau}\displaystyle\sum_{j=1}^p\mathbb{E}_{\mathbf{\theta}_\star^{(1)}}\left[\textsf{Var}_{\theta_\star^{(1)}}\left(\sum_{k=1}^p\Delta_{jk}B_{ik}(X_j,X_k)\vert X_{-j}\right)\right].
\end{multline}
By the comparison Lemma \ref{lemlip}
\begin{multline*}
\label{eqn 2:lemma 5}
\left|\mathbb{E}_{\mathbf{\theta}_\star^{(2)}}\left[\textsf{Var}_{\theta_\star^{(1)}}\left(\sum_{k=1}^p\Delta_{jk}B_{ik}(X_j,X_k)\vert X_{-j}\right)\right] - \mathbb{E}_{\mathbf{\theta}_\star^{(1)}}\left[\textsf{Var}_{\theta_\star^{(1)}}\left(\sum_{k=1}^p\Delta_{jk}B_{ik}(X_j,X_k)\vert X_{-j}\right)\right]\right|\\
\leq c_0^3\left(\sum_{k=1}^p|\Delta_{jk}|\right)^2\sum_{k=1}^p|\theta_{\star jk}^{(1)}-\theta_{\star jk}^{(2)}|\leq c_0^3 b\left(\sum_{k=1}^p|\Delta_{jk}|\right)^2,
\end{multline*}
which implies that 
\[\mathcal{V}^1\left(\tau,\Delta\right)\geq \left(2\rho_1-\frac{64}{\tau}(\tau-\tau_\star)_+s_1c_0^3 b\right) \|\Delta\|^2_2 + \frac{1}{\tau}\displaystyle\sum_{t=1}^\tau\sum_{j=1}^p\sum_{k,k^{\prime}=1}^p\Delta_{jk}\Delta_{jk^{\prime}}W^{(t)}_{jkk^{\prime}} .\]
Given that on $\Tau_+$, $128(\tau-\tau_\star)s_1c_0^3b\leq\rho_1\tau$, it follows that for all $\tau\in\mathcal{T}$,
\begin{equation}
\label{eqn 3:lemma 5}
\mathcal{V}^1\left(\tau,\Delta\right)\geq\frac{3}{2}\rho_1\|\Delta\|^2_2 + \frac{1}{\tau}\displaystyle\sum_{t=1}^\tau\sum_{j=1}^p\sum_{k,k^{\prime}=1}^p\Delta_{jk}\Delta_{jk^{\prime}}W^{(t)}_{jkk^{\prime}}
\end{equation}
Set $Z^\tau_{jkk^{\prime}}\stackrel{\text{def}}{=}\frac{1}{\tau}\displaystyle\sum_{t=1}^\tau W^{(t)}_{jkk^{\prime}}$. We conclude from equation ~\eqref{eqn 3:lemma 5} that if for some $\Delta\in\mathbb{C}_1\setminus\left\{0\right\}$, and for some $\tau\in\Tau$,
\begin{equation}
\label{eqn 4:lemma 5}
\mathcal{V}^1\left(\tau,\Delta\right)\leq \rho_1\|\Delta\|^2_2
\end{equation}
then $$ \displaystyle\sum_{j=1}^p\sum_{k,k^{\prime}=1}^p\Delta_{jk}\Delta_{jk^{\prime}}Z^{(\tau)}_{jkk^{\prime}}\leq-\frac{\rho_1}{2}\|\Delta\|^2_2\mbox{.}$$
But on the other hand, using the fact that $\Delta\in\mathbb{C}_1$,
\begin{eqnarray}
\displaystyle\sum_{j=1}^p\sum_{k,k^{\prime}=1}^p\Delta_{jk}\Delta_{jk^{\prime}}Z^{(\tau)}_{jkk^{\prime}} &\geq& -\left(\displaystyle\sup_{j,k,k^{\prime}}|Z^{(\tau)}_{jkk^{\prime}}|\right)\left(\displaystyle\sum_{i=1}^p\sum_{k=1}^p|\Delta_{ik}|\right)^2 \nonumber \\
                                                                                                                        &\geq& -\left(\displaystyle\sup_{j,k,k^{\prime}}|Z^{(\tau)}_{jkk^{\prime}}|\right) 4\|\Delta\|^2_1 \nonumber \\
																																																												&\geq& -64s_1\left(\displaystyle\sup_{j,k,k^{\prime}}|Z^{(\tau)}_{jkk^{\prime}}|\right)\|\Delta\|^2_2\mbox{.} \nonumber
\end{eqnarray}
Therefore if there exists a non-zero $\Delta\in\mathbb{C}_1$ and $\tau\in\Tau$ such that equation ~\eqref{eqn 4:lemma 5} holds then $ \left(\displaystyle\sup_{j,k,k^{\prime}}|Z^{(\tau)}_{jkk^{\prime}}|\right)\geq(\rho_1/s_1)(1/128)$. But by Hoeffding's inequality and a union-sum bound, 
$$ \mathbb{P}\left[\displaystyle\sup_{j,k,k^{\prime}}|Z^{(\tau)}_{jkk^{\prime}}|\geq\frac{\rho_1}{128s_1}\right]\leq 2\exp\left(3\log p-\frac{\tau\rho^2_1}{2^9c_0^2s^2_1}\right)\leq\frac{2}{p}\mbox{,}$$ since for $\tau\in\Tau$, $\tau\geq 2^{11}c_0^2s^2_1\rho_1^{-2}\log p$.\qed
\end{proof}

\begin{Lemma}\label{lemlip}
Let $(\mathsf{Y},\mathcal{A},\nu)$ be a measure space where $\nu$ is a finite
measure. Let $g_1,g_2,f_1,f_2:\;\Yset\to\rset$ be  bounded measurable functions. Set $Z_{g_i}\eqdef \int_{\textsf{Y}} e^{g_i(y)}\nu(dy)$, $i\in\{1,2\}$. Then
\begin{multline*}
\left|\frac{1}{Z_{g_1}}\int f_1(y)e^{g_1(y)}\nu(dy)-\frac{1}{Z_{g_2}}\int f_2(y)e^{g_2(y)}\nu(dy)\right| \\
\leq \|f_2-f_1\|_\infty + \frac{1}{2}\textsf{osc}(g_2-g_1)\left(\textsf{osc}(f_1)+\textsf{osc}(f_2)\right),
\end{multline*}
where $\|f\|_\infty=\sup_{x\in \mathsf{Y}}|f(x)|$, and $\textsf{osc}(f)\eqdef\sup_{x,y\in\mathsf{Y}}|f(x)-f(y)|$ is the oscillation of $f$.
\end{Lemma}
\begin{proof}
The proof follows from \cite{atc}~Lemma 3.4.
\end{proof}
\section{Different Methods of Missing Data Imputation for the Real Data Application}
In the main paper we replaced the missing votes by the value (yes/no) of that member's party
majority position on that particular vote. Here we employed two other missing data imputation techniques viz. (i) replacing all missing values by the value (yes/no) representing the winning majority on that bill and (ii) replacing the missing value of a Senator by the value that the majority of the opposite party voted on that particular bill. The estimated change-point obtained following these two imputation methods are not much different . The imputation technique (i) results in a estimated change-point at January 19, 1995 and the technique (ii) yields estimated change-point at January 17, 1995 respectively. The change-point estimate we obtained in the main paper was January 17, 1995. Clearly there is not much difference between the different imputation techniques and Fig. \ref{fig:imput} also conveys the same message.

\begin{figure}
\centering
\begin{subfigure}{.5\textwidth}
  \centering
  \includegraphics[width=.7\linewidth]{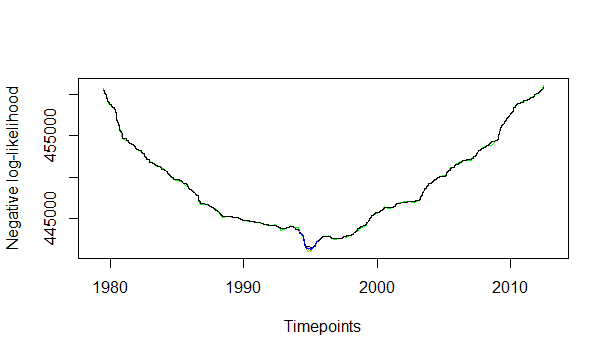}
  \label{fig:sub1}
\end{subfigure}%
\begin{subfigure}{.5\textwidth}
  \centering
  \includegraphics[width=.7\linewidth]{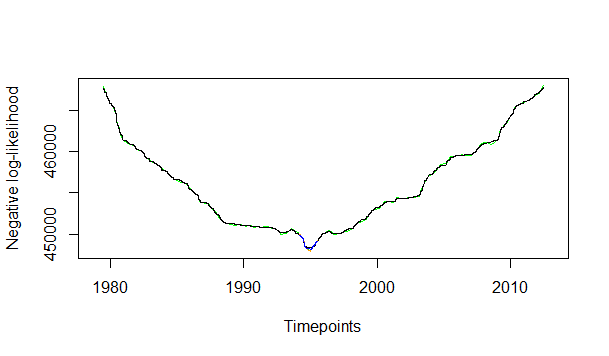}
  \label{fig:sub2}
\end{subfigure}
\caption{Estimated Change-points via imputation technique (i) and (ii) respectively}
\label{fig:imput}
\end{figure}

\end{document}